\documentclass[12pt]{article}

\usepackage{amsmath, amssymb, amsthm}
\usepackage{bm}
\usepackage{mathrsfs}
\usepackage{mathtools}
\usepackage{enumitem}
\usepackage{geometry}
\usepackage{setspace}
\usepackage{hyperref}

\newtheorem{theorem}{Theorem}
\newtheorem{lemma}{Lemma}
\newtheorem{proposition}{Proposition}
\newtheorem{corollary}{Corollary}

\theoremstyle{definition}
\newtheorem{definition}{Definition}
\newtheorem{assumption}{Assumption}

\theoremstyle{remark}
\newtheorem{remark}{Remark}

\newcommand{\R}{\mathbb{R}}
\newcommand{\N}{\mathbb{N}}

\title{\bf Feedback Linearization of Hyperbolic PDEs\\ with  Volterra 
Nonlinearities}
\author{Miroslav Krstic\thanks{Department of Mechanical and Aerospace Engineering, 
University of California San Diego, 
La Jolla, CA 92093-0411, USA,  
\texttt{mkrstic@ucsd.edu}}}
\date{}

\begin{document}

\maketitle

\begin{center}
{\em Invited paper in honor of Professor Alberto Isidori's
85th birthday }  
\end{center}

\bigskip

\begin{abstract}
Alberto Isidori's framework of geometric nonlinear control, and particularly of feedback linearization, is the inspiration behind PDE backstepping: apply a transfromation of the state to cast the plant into a canonical form, bring all the non-canonical effects within the ``span'' of (boundary) control, and close the design with a feedback that makes the closed loop evolve in accordance with well-studied stable dynamics. The specificity of this approach is that, for PDEs, there is not one canonical form (like Brunovsky for ODEs) but the canonical forms are PDE-class-specific. When conducting this process for nonlinear PDEs, where the ``transformation of the state'' is performed using a nonlinear Volterra series indexed by the spatial variable, enormous technical challenges arise. One has to deal with kernels governed by PDEs on simplex domains growing in dimension to infinity, capture the growth rates of these kernels of the ``direct transformation,'' and conduct the same for the ``inverse transformation'' without directly studying its Volterra kernels. So far, this agenda has been executed only once, two decades ago: for parabolic PDEs by Vazquez and Krstic [Automatica, 2008]. Generalization attempts have not followed because of the immense complexity involved in feedback-linearizing nonlinear PDEs. 

In this paper, dedicated to Professor Isidori, we convert the PDE feedback-linearizing methodology of 2008 from the parabolic to a hyperbolic class and, for a transport-adapted subclass of Chen–Fliess series, construct controllers without kernel PDEs. 
\end{abstract}


\newpage

\section{Introduction}\label{sec:introduction}

\paragraph{Motivation.} It is common that feedback linearization \cite{isidori1989nonlinear}, a concept that sits at the foundation of control theory, be regarded as a {\em method} for nonlinear control design. In consequence, feedback linearization is occasionally unduly viewed as canceling useful nonlinearities or lacking robustness. 

The truth about feedback linearization, or more precisely feedback linearizability, is that it is a {\em property} of a nonlinear system. Numerous design choices exist for systems that possess such a desirable property. One should view feedback linearizability as an extension to concepts like the Hartman-Grobmann theorem, which seeks linearization (in the form of local topological equivalence; homeomorphic, not necessarily diffeomorphic) of input-free systems by coordinate change only. Feedback linearization, the more powerful `cousin' to Hartman-Grobmann, deals with a more general system class and is focused on linearization by coordinate change so that all that remains nonlinear falls within the span of control. Feedback that opts for linearity in closed loop is just one design option, not the general recipe. 

Backstepping, a design method, emerged for application to both systems that are feedback linearizable and some that are not. In a pioneering paper with Byrnes \cite{byrnes1989new}, Alberto Isidori introduced one of the simultaneous variants of the backstepping design ideas, which appeared in or around 1989. 

The success of backstepping for nonlinear ODEs led to the idea of using this approach for nonlinear PDEs as well. The modest start of this approach, in a setting fully free of spatial discretization to high-dimensional ODEs, was in 2001 for a small class of linear parabolic PDEs \cite{boskovic2001boundary}. Nonlinear PDEs long evaded PDE backstepping until the 2008 two-part paper \cite{vazquez2008volterra1,vazquez2008volterra2}. The papers \cite{vazquez2008volterra1,vazquez2008volterra2} expose such a degree of technical complexity to feedback linearization in infinite dimension, i.e., for nonlinear PDEs, that \cite{vazquez2008volterra1,vazquez2008volterra2} has remained for nearly two decades the sole result on feedback linearization for a nonlinear PDE class. 

In this paper, dedicated to the pioneer and most substantive and comprehensive contributor to geometric nonlinear control, including feedback linearization, Professor Isidori, we take the first step since 2008 \cite{vazquez2008volterra1,vazquez2008volterra2} in advancing feedback linearization to another nonlinear PDE class.

\paragraph{Problem formulation.}
The paper's purpose is to develop a feedback linearization procedure for a nonlinear first-order hyperbolic PDE with a spatial Volterra nonlinearity. The starting point is a transport equation perturbed by a nonlinear operator expressed as an infinite Volterra series in the spatial variable. Specifically, we consider feedback linearization for a class of nonlinear first-order hyperbolic partial differential equations of the form
\begin{equation}\label{eq:plant}
u_t(x,t) = u_x(x,t) + F[u](x,t), \quad x \in [0,1), \ t \ge 0,
\end{equation}
subject to the boundary condition
\begin{equation}\label{eq:boundary}
u(1,t) = U(t),
\end{equation}
and the initial condition
\begin{equation}\label{eq:initial}
u(x,0) = u_0(x) \in L^2(0,1).
\end{equation}
Here, $F[u](x,t)$ is assumed to be a Volterra series in the spatial variable whose expansion begins at second order:
\begin{equation}\label{eq:VolterraF}
F[u](x,t)
=
\sum_{n=2}^{\infty}
\int_{T_n(x)}
f_n(x,\xi_1,\dots,\xi_n)
\prod_{i=1}^{n} u(\xi_i,t)\,
d\xi_n \cdots d\xi_1,
\end{equation}
where the kernels $f_n$ are measurable and essentially bounded on the simplices
\begin{equation}\label{eq:simplex}
T_n(x) := \{(\xi_1,\dots,\xi_n)\in\mathbb{R}^n:\ 0\le \xi_n \le \cdots \le \xi_1 \le x\}.
\end{equation}
We exclude the linear Volterra term $\int_0^x f_1(x,\xi)u(\xi,t)\,d\xi$ without loss of generality: by applying a preliminary step of (linear) backstepping transformation, using \cite{krstic2008backstepping}, such a term can be eliminated from the model. This is performed as follows. We transform the model $v_t(x,t)=v_x(x,t)+G[v](x,t)$ to $u_t(x,t)=u_x(x,t)+F[u](x,t)$ using the linear backstepping transform $u(x,t)=v(x,t)-\int_0^x p(x,y)\,v(y,t)\,dy$, whose inverse is  $v(x,t)=u(x,t)+\int_0^x q(x,y)\,u(y,t)\,dy$, and whose kernels $p,q$ are defined by the Volterra integral equations $p(x,y)
=-\int_y^x p(x,s)\,g_1(s,y)\,ds
+\int_y^x g_1(s,y)\,ds$ and
$q(x,y)=p(x,y)
+\int_y^x p(x,s)\,q(s,y)\,ds$ over $ 0\le y\le x\le 1$. 
For $n\ge 2$, the resulting kernels $f_n$, or the quadratic and higher-order Volterra terms, are given by 
\begin{equation}\label{eq:fn_def}
f_n(x,\eta_1,\ldots,\eta_n)=h_n(x,\eta_1,\ldots,\eta_n)-\int_{\eta_1}^{x}
p(x,y)\,h_n(y,\eta_1,\ldots,\eta_n)\,dy,
\end{equation}
where
\begin{equation}\label{eq:hn_def}
h_n(x,\eta_1,\ldots,\eta_n)=g_n(x,\eta_1,\ldots,\eta_n)+\sum_{\varnothing\ne S\subseteq\{1,\ldots,n\}}
\int_{\Omega_S(x,\eta)}
g_n(x,\tilde\zeta_1,\ldots,\tilde\zeta_n)
\prod_{i\in S}q(\zeta_i,\eta_i)
\prod_{i\in S}d\zeta_i 
\end{equation}
and the index $S$ denotes a subset of $\{1,\ldots,n\}$. The case $S=\varnothing$ gives $h_n(x,\eta_1,\ldots,\eta_n)=g_n(x,\eta_1,\ldots,\eta_n)$, whereas  the sum is taken over all nonempty  $S$,  with integration over
\begin{equation}\label{eq:Omega_def}
\Omega_S(x,\eta)=\{(\zeta_i)_{i\in S}:\eta_i\le \zeta_i\le x\ \text{for all }i\in S,\ 0\le\tilde\zeta_n\le\cdots\le\tilde\zeta_1\le x\},
\end{equation}
and with $\tilde\zeta_i=\zeta_i$ for $i\in S$ and $\tilde\zeta_i=\eta_i$ for $i\notin S$. As long as the original linear kernel $g_1$ is bounded on $0\le y\le x\le 1$ and the original nonlinear sequence $\{g_n\}_{n\ge 2}$ satisfies a factorial bound, it can be shown that the resulting $\{f_n\}_{n\ge 2}$ satisfies its own factorial bound of the same shape, with  constants modified by $\|g_1\|_\infty$. 


To recap, we leave the linear term out of $F[u]$ because, in a paper on infinite-dimensional feedback {\em linearization}, it is manifestly meaningful to focus on nonlinearities with quadratic and higher terms only, and not contaminate the analysis with linear issues. Shaping the linear behavior is a problem distinct from linearization by feedback.

The goal is to construct a nonlinear change of variables, also given as a Volterra series, that maps the closed-loop system into a pure transport equation with homogeneous boundary condition, for which stability properties are already known. Concretely, we construct a nonlinear Volterra transformation
\begin{equation}\label{eq:transformation}
w(x,t) = u(x,t) - K[u](x,t),
\end{equation}
where $K[u](x,t)$ is of the same Volterra form as in \eqref{eq:VolterraF}, together with the boundary feedback law
\begin{equation}\label{eq:control}
U(t) = K[u](1,t),
\end{equation}
such that the closed-loop system is mapped into the transport equation
\begin{equation}\label{eq:target}
w_t(x,t) = w_x(x,t), \quad w(1,t) = 0.
\end{equation}
After deriving the kernel equations associated with the transformation \eqref{eq:transformation}, we analyze the direct and inverse mappings generated by $K[u]$, and  establish local exponential stability of the closed-loop system in the $L^2$-norm of the original state $u$.

\paragraph{Summary of the paper's approach and results.}
The construction of the backstepping operator $K$, which is used in both the model-linearizing state transformation and in the closed-loop-linearity-ensuring feedback, proceeds by substituting the transformation $I-K$ into the plant and enforcing that the transformed state satisfy the target transport equation. This yields a ``hierarchy'' of linear first-order PDEs for the transformation kernels, indexed by the order of the Volterra expansion. Each kernel equation is posed on a simplex domain and has an inflow boundary condition. The equations are coupled through lower-order kernels but retain a transport structure, which allows them to be solved explicitly along characteristics. This explicit representation is the basis for all subsequent estimates.

The main technical step is to control the growth of the kernels with respect to the order of the expansion. By integrating the characteristic representation and bounding the multilinear coupling terms, a recursive $L^2$ inequality is obtained for the sequence of kernel norms. This recursion depends only on the bounds of the original nonlinear kernels and lower-order terms of the same sequence. The recursion is then majorized by an equality recursion. By working with factorially weighted coefficients, one obtains a convolution estimate that ensures convergence of the associated gain functions and guarantees that the resulting transformation is well-defined on a neighborhood of the origin in $L^2$. This provides a controlled way to track how the kernel magnitudes grow with the order and ensures that the resulting transformation is well-defined on a neighborhood of the origin in $L^2$.

With the kernel bounds in place, the direct transformation is shown to be locally well-defined and continuous. The inverse transformation is obtained by solving a fixed-point equation, where the nonlinear part is shown to be locally Lipschitz with a gain that can be made strictly less than one by restricting to a sufficiently small ball. This ensures that the inverse exists and is unique via a contraction argument. The same estimates also provide quantitative control of the inverse mapping.

Finally, the closed-loop system is analyzed. The transformed state evolves according to a stable transport equation, yielding exponential decay in the $L^2$ norm. The relationship between the original and transformed states is then used to transfer this decay back to the original variables. The bounds established for the transformation and its inverse ensure that the solution remains within the region where the analysis applies for all time, resulting in local exponential stability of the closed-loop system.

\paragraph{A bonus: Specialization to a Chen--Fliess series subclass.}
To obtain a coefficient-level construction of the controller, we restrict, in the final part of the paper, the Volterra-series nonlinearity to a special transport-adapted subclass of Volterra series whose kernels admit a Chen--Fliess series representation.
In this subclass, the kernels have a particular, essentially separable and multinomial form. The consequence is computational rather than analytical: the recursive construction of the backstepping kernels is carried out through scalar coefficient equations on $[0,1]$, rather than through characteristic formulas on growing simplices. The convergence analysis is not eliminated; it is transferred to a divided-power coefficient algebra.

\paragraph{Relation to work for parabolic nonlinear PDEs.}
The present work parallels the nonlinear PDE feedback-linearization program introduced by Vazquez and Krstic \cite{vazquez2008volterra1,vazquez2008volterra2} for parabolic PDEs, but for a different class of systems and with correspondingly different analytical features. In the parabolic setting, the kernel equations are second-order PDEs posed on growing simplices, and their analysis relies on energy estimates and successive differentiation. In contrast, the hyperbolic kernels here satisfy first-order transport equations, which admit explicit representations along characteristics. This allows the kernel bounds to be obtained by direct Gr\"onwall integration of a recursive ODE inequality rather than by the parabolic differentiation-and-energy machinery. Despite this structural simplification, the coupling across orders and the growth of simplex dimension remain the central difficulty in both settings. The use of Volterra-series representations, hierarchical kernel construction, and majorant arguments for convergence are inherited directly from the parabolic case, although the technical execution differs. A further methodological distinction concerns the inverse transformation: \cite{vazquez2008volterra1,vazquez2008volterra2} establish invertibility via Boyd--Chua formal series inversion \cite{boyd1985fading}, which yields existence of the inverse Volterra series but no quantitative estimates on its kernels. The present work replaces that construction with a direct contraction-mapping argument on $L^2$, yielding an explicit Lipschitz bound on the inverse and an explicit radius of invertibility. In both works, the outcome is a locally defined nonlinear transformation with a well-defined inverse and a closed-loop system that is exponentially stable in the $L^2$ norm of the original state. The characterization of the region of attraction is achieved through small-gain-type conditions in both settings; in the present hyperbolic case, these conditions are made quantitative through the explicit gain function $\ell(s)$ governing the contraction.

\paragraph{Relation to feedback linearization for ODEs.} Denote by $\mathcal{T}$ the transformation  introduced in \eqref{eq:transformation} but now named and written as an operator:
\begin{equation}\label{eq:T-trans}
\mathcal{T}[u]:= (I-K)[u],
\end{equation}
where $I$ denotes the identity operator on $L^2(0,1)$. It is with this transformation, and feedback \eqref{eq:boundary}, \eqref{eq:control}, that we convert the nonlinear PDE \eqref{eq:plant} into the linear target PDE \eqref{eq:target}. This transformation $\mathcal{T}[u]$ is the analog of a diffeomorphic transformation $T(x)$ used for ODE systems of the form $\dot x = f(x)+g(x)u$. In Section \ref{sec:inverse} we show that $\mathcal{T}$ is invertible, and in Section \ref{sec:diffeo} that both $\mathcal{T}$ and $\mathcal{T}^{-1}$ are differentiable in the Fr\'echet sense, making what we present in this paper a full nonlinear PDE analog to diffeomorphisms $T(\cdot)$ used in ODE feedback linearization. 

Just as $T(\cdot)$ transforms $\dot x = f(x) + g(x)u$ into Brunovsky form by a change of coordinates that brings all the nonlinearities within the span of the control, the operator $\mathcal{T}$ alone maps the $u$-plant \eqref{eq:plant} into the $w$-system \eqref{eq:target} on $[0,1)$ and into the boundary relation $w(1,t) = U(t) - K[u](1,t)$, bringing all nonlinearities of the plant within the span of the boundary input $U(t)$.

For PDEs, no general control theory can exist---all analysis and design for PDEs is class specific. But we can establish an even closer analogy between control-affine ODE systems $\dot x=f(x) +g(x)u$ and the PDE class \eqref{eq:plant}, \eqref{eq:boundary} than merely through diffeomorphisms. For the control-affine PDE class \eqref{eq:plant}, \eqref{eq:boundary}, the role of the drift vector field $f(\cdot)$ is played by the nonlinear operator $(\partial_x + F)[\cdot]$, whereas the role of the input vector field $g(\cdot)$ is played by the Dirichlet boundary input operator (which is an unbounded linear operator). 

Viewed in coordinate-free terms, feedback linearization of the plant \eqref{eq:plant}--\eqref{eq:boundary} reduces to a single ``PDE-like'' equation for the backstepping operator $K$, namely
\begin{equation}\label{eq:cohomological}
L_F\, K = F,
\end{equation}
where $L_F$ denotes the directional Fr\'echet derivative of an operator with respect to its $L^2(0,1)$-valued argument, taken in the direction specified by the plant Volterra operator $F$. The operator $K$ is thus characterized as the object whose directional Fr\'echet derivative along $F$ reproduces $F$ itself---a cohomological equation in the infinite-dimensional state space $L^2(0,1)$. The role of the input vector field $g(x)$ in the ODE analogy is played here by the Dirichlet boundary operator $\mathrm{ev}_{x=1}$, through which the feedback $U(t) = K[u](1,t)$ routes the plant's nonlinearity into the span of the boundary control and enforces the target's homogeneous boundary condition $w(1,t)=0$. This is the infinite-dimensional analog of the matching PDE $L_f T = AT$ satisfied by the diffeomorphism $T(x)$ in ODE feedback linearization for the companion matrix $A$ of the Brunovsky form. The factorial growth bound $\|f_n\|_\infty \leq n!\,D_f/\rho_f^{\,n-1}$ on the Volterra kernels of $F$ plays, in this infinite-dimensional setting, a role reminiscent of the Hunt--Su--Meyer \cite{hunt1983global} rank-and-involutivity conditions for ODE feedback linearization: it is a checkable property of the plant data $F$ alone---involving neither the unknown transformation $K$ nor iterated constructions like Lie brackets---whose verification certifies that the feedback-linearizing transformation exists, converges, and yields a $C^1$-diffeomorphism on a neighborhood of the origin in $L^2$.

\paragraph{Relation to input-output representation of nonlinear systems.} The Volterra-series structure of the plant nonlinearity $F[u]$ is not foreign to nonlinear ODE sysetms with control. The problem of representing the input--output map of a nonlinear ODE by a Volterra series --- the natural nonlinear generalization of the convolution representation $y=w*u$ of finite-dimensional linear systems --- was opened in the foundational work of d'Alessandro, Isidori and Ruberti \cite{dAlessandroIsidoriRuberti1974}, who, for bilinear systems $\dot{x}=Ax+Nxu+Bu$, $y=Cx$, obtained the Volterra kernels in closed form as products of matrix exponentials interleaved with $N$, and characterized which kernel sequences are realizable by such a bilinear state-space model in terms of a factorizability/rationality condition. Brockett \cite{Brockett1976} extended the construction to the linear-analytic (control-affine analytic) class $\dot{x}=f(t,x)+u\,g(t,x)$, $y=h(t,x)$, computing the kernels by Picard iteration after eliminating the drift through the unforced flow and proving uniform convergence under a small-input bound via Picard/Gronwall estimates. Lesiak and Krener \cite{1101898} subsequently established that, for control-affine analytic systems with zero initial condition, the input--output map $u\mapsto y$ admits a Volterra-series representation with kernels determined recursively from $f,g,h$ along the unforced flow, and convergence proved by Cauchy estimates on polydiscs of analyticity yielding factorial-times-geometric bounds. Reading time as a spatial variable, the same construction yields a class of Volterra operators $u(\cdot)\mapsto F[u](\cdot)$ of precisely the form \eqref{eq:VolterraF}. The factorial-growth bound $\|f_n\|_\infty \leq n!\,D_f/\rho_f^{\,n-1}$ is the natural analog, in our setting, of the Cauchy estimates governing convergence of the Krener Volterra series in \cite{1101898}. The PDAE example of Section \ref{sec:example} belongs to this Krener class, with $f\equiv 0$, $g\equiv 1$, $h(z)=z^2/2$.

\paragraph{Organization of the paper.}
Section \ref{sec:backstepping} introduces the Volterra backstepping transformation and derives the matching condition for the kernels. Section \ref{sec:kernels} formulates the kernel equations and their characteristic integral representation. Section \ref{sec:direct} establishes convergence of the direct transformation via $L^2$ kernel bounds and a majorant construction. Section \ref{sec-quant} provides a quantitative lower bound on the convergence radius of the gain functions. Section \ref{sec:inverse} develops the inverse transformation through a contraction mapping argument. Section \ref{sec:closed-loop} presents the main result on local exponential stability of the closed-loop system. The well-posedness in closed loop is shown in Section \ref{sec-exist}. Section \ref{sec:diffeo} discusses the differentiability and invertibility properties of the Volterra-based backstepping transformation: an analogy to the diffeomorphisms in ODE feedback linearization. Section \ref{sec:example} illustrates the theory on a nonlinear PDAE example. 

Finally, Section \ref{sec:pde-free} shows that, for plants whose Volterra kernels admit a real-analytic Chen--Fliess gap-basis expansion with geometrically decaying coefficients, the infinite cascading set of simplex transport PDEs underlying the backstepping controller is replaced by a triangular cascade of scalar ODEs on $[0,1]$ solvable by quadrature, with the resulting gap-basis series proved to converge absolutely and uniformly on $T_n(1)$ with $n$-uniform bounds that align with the Volterra convergence analysis of Section \ref{sec:direct}. 

\section{Backstepping Transformation}\label{sec:backstepping}

We seek a Volterra transformation of the form \eqref{eq:transformation}, with $K[u]$ defined by
\begin{equation}\label{eq:K_series}
K[u](x,t)
=
\sum_{n=2}^{\infty}
\int_{T_n(x)}
k_n(x,\xi_1,\dots,\xi_n)
\prod_{i=1}^{n} u(\xi_i,t)\,
d\xi_n \cdots d\xi_1,
\end{equation}
where the kernels $k_n$ are to be determined.

Substituting \eqref{eq:transformation} into \eqref{eq:plant} and using \eqref{eq:control}, we obtain
\begin{eqnarray}\label{eq:w_evolution}
w_t(x,t)
&=&
u_t(x,t) - \partial_t K[u](x,t) \nonumber\\
&=&
u_x(x,t) + F[u](x,t) - \partial_t K[u](x,t).
\end{eqnarray}
Similarly,
\begin{equation}\label{eq:w_x}
w_x(x,t)
=
u_x(x,t) - \partial_x K[u](x,t).
\end{equation}
Imposing \eqref{eq:target}, that is, $w_t = w_x$, and combining \eqref{eq:w_evolution}--\eqref{eq:w_x}, yields
\begin{equation}\label{eq:matching}
\partial_t K[u](x,t)
-
\partial_x K[u](x,t)
=
F[u](x,t).
\end{equation}
The boundary condition \eqref{eq:target} is satisfied by construction from \eqref{eq:control}, since
\begin{equation}\label{eq:w_boundary}
w(1,t)
=
u(1,t) - K[u](1,t)
=
U(t) - K[u](1,t)
=
0.
\end{equation}
Equation \eqref{eq:matching}, together with the Volterra representations \eqref{eq:VolterraF} and \eqref{eq:K_series}, determines the kernel equations for $k_n$ by identification of terms of equal order in $u$.

\section{Kernel PDEs and Integral Representation}\label{sec:kernels}

From \eqref{eq:matching}, with $K[u]$ given by \eqref{eq:K_series} and $F[u]$ by \eqref{eq:VolterraF}, identification of terms of equal order in $u$ yields, for each $n\ge 2$,
\begin{eqnarray}
\label{eq:kernel_pde}
\partial_x k_n(x,\xi_1,\dots,\xi_n)
+
\sum_{i=1}^{n} \partial_{\xi_i} k_n(x,\xi_1,\dots,\xi_n)
&=&
- f_n(x,\xi_1,\dots,\xi_n)
\nonumber\\
&& +\, \sum_{m=2}^{n} B_n^m\!\left[k_{n-m+1}, f_m\right](x,\xi_1,\dots,\xi_n),
\end{eqnarray}
for $(x,\xi_1,\dots,\xi_n)\in T_n(1)$, where, with the convention $\xi_0:=x$, the integral operators $B_n^m$ are defined as follows.

For $2\le m\le n$,
\begin{eqnarray}
\label{eq:Cnm_def}
&& B_n^m\!\left[k_{n-m+1}, f_m\right](x,\xi_1,\dots,\xi_n) 
\nonumber\\
&&\quad =
\sum_{j=1}^{n-m+1}
\int_{\xi_j}^{\xi_{j-1}}
D_j^{\,n-m+1,m}\!\Big[
k_{n-m+1}(x,\xi_1,\dots,\xi_{j-1},s,\xi_j,\dots,\xi_{n-m})
\nonumber\\
&&\qquad\qquad\qquad\qquad\qquad\;\;\;\times f_m(s,\xi_{n-m+1},\dots,\xi_n)
\Big]\, ds.
\end{eqnarray}
The operator $D_j^{\,p,m}$ appearing in \eqref{eq:Cnm_def} is a permutation-summation operator, defined for a function $g(x,\zeta_1,\dots,\zeta_{p+m-1})$ and indices $p\ge 1$, $m\ge 2$, $1\le j\le p$, by
\begin{eqnarray}
\label{eq:D_def}
&& D_j^{\,p,m}\!\big[g(x,\zeta_1,\dots,\zeta_{p+m-1})\big]
\nonumber\\
&&\quad =
\sum_{(\gamma_1,\dots,\gamma_{p+m-1-j})\,\in\,\mathcal{P}_{p-j}(\zeta_{j+1},\dots,\zeta_{p+m-1})}
g(x,\zeta_1,\dots,\zeta_{j-1},\zeta_j,\gamma_1,\dots,\gamma_{p+m-1-j}),
\end{eqnarray}
where $\mathcal{P}_{p-j}(\zeta_{j+1},\dots,\zeta_{p+m-1})$ denotes the set of all ordered $(p+m-1-j)$-tuples formed as follows: the first $p-j$ entries are any $p-j$ elements of $\{\zeta_{j+1},\dots,\zeta_{p+m-1}\}$ taken in their original order, and the remaining $m-1$ entries are the leftover elements of the same set, also taken in their original order.

The RHS of \eqref{eq:kernel_pde} consists of two types of contributions. The forcing term $-f_n$ is given data. The terms $B_n^m[k_{n-m+1},f_m]$ for $m=2,\dots,n$ couple $k_n$ to the lower-order kernels $k_2,\dots,k_{n-1}$ through integrals against $f_m$; they act as forcing terms in the recursive solution, since at step $n$ the kernels $k_2,\dots,k_{n-1}$ are already known. (For $n=2$, the sum in \eqref{eq:kernel_pde} is empty, and the equation reduces to $\partial_x k_2 + \sum_{i=1}^{2}\partial_{\xi_i}k_2 = -f_2$.)

The boundary condition for \eqref{eq:kernel_pde} is obtained from the transport structure. Since $K[u]$ is Volterra in $x$, the inflow boundary corresponds to $\xi_n=0$, and one has
\begin{equation}\label{eq:kernel_bc}
k_n(x,\xi_1,\dots,\xi_{n-1},0)=0,
\qquad 0\le \xi_{n-1}\le \cdots \le \xi_1 \le x \le 1.
\end{equation}
This is the only boundary condition required, in contrast to the parabolic case where additional conditions are imposed on the face $\xi_1=x$.

Equations \eqref{eq:kernel_pde}--\eqref{eq:kernel_bc} define a first-order linear transport problem for each kernel $k_n$, forced by $f_n$ and coupled to the lower-order kernels $k_2,\dots,k_{n-1}$ through $B_n^m$, $m=2,\dots,n$. Because the coupling is only to lower orders, \eqref{eq:kernel_pde}--\eqref{eq:kernel_bc} can be solved recursively starting from $n=2$.

The LHS of \eqref{eq:kernel_pde} can be integrated along characteristics. For $(x,\xi_1,\dots,\xi_n)\in T_n(1)$, define the characteristic curve
\begin{equation}\label{eq:characteristic}
s \mapsto (x-s,\xi_1-s,\dots,\xi_{n-1}-s,\xi_n-s),
\end{equation}
which intersects the boundary $\xi_n=0$ at $s=\xi_n$. Integration of \eqref{eq:kernel_pde} along this curve, together with the inflow condition \eqref{eq:kernel_bc}, yields the explicit representation
\begin{align}\label{eq:kernel_integral}
k_n(x,\xi_1,\dots,\xi_n)
= -\int_{0}^{\xi_n}
\Bigg[f_n -
\sum_{m=2}^{n} B_n^m\!\left[k_{n-m+1}, f_m\right]
\Bigg]
(x-\xi_n+s,
\xi_1-\xi_n+s,&
\nonumber\\
\dots,\xi_{n-1}-\xi_n+s,s)\, ds.&
\end{align}
At each order $n\ge 2$, equation \eqref{eq:kernel_integral} is an explicit formula for $k_n$ in terms of $f_n$ and the lower-order kernels $k_2,\dots,k_{n-1}$, which are treated as known data entering through the $B_n^m$ terms. Starting from $k_2$, for which \eqref{eq:kernel_integral} reduces to $k_2(x,\xi_1,\xi_2) = -\int_0^{\xi_2} f_2(x-\xi_2+s,\xi_1-\xi_2+s,s)\,ds$, the higher-order kernels are obtained iteratively.

\section{Direct Transformation: Convergence}\label{sec:direct}
To proceed into the convergence analysis, we make a key assumption on the kernels $f_n$. 
\begin{assumption}\label{ass:weighted_growth}
The kernels $f_n$ satisfy $f_n \in C^0(T_n(1))$ for each $n\ge 2$, and there exist positive constants $D_f$ and $\rho_f$ such that
\begin{equation}\label{eq:an_def}
\|f_n\|_{L^\infty(T_n(1))} \le \frac{n!\, D_f}{\rho_f^{\,n-1}},
\qquad n \ge 2.
\end{equation}
\end{assumption}

\begin{remark}\label{rem:assumption_scope}
Since $\mathrm{vol}(T_n(1)) = 1/n!$, the bound \eqref{eq:an_def} implies the $L^2$ estimate
\begin{equation}\label{eq:fn_L2_bound}
\|f_n\|_{L^2(T_n(1))}^2 \le \frac{n!\, D_f^{\,2}}{\rho_f^{\,2(n-1)}},
\qquad n\ge 2,
\end{equation}
which is of the factorial-growth type handled by the Gain Bound Theorem in  \cite[Theorem 1]{vazquez2008volterra2}. By that theorem, the Volterra series \eqref{eq:VolterraF} converges in $L^2(0,1)$ whenever $\|u\|_{L^2(0,1)}^2 < \rho_f^{\,2}$.
Special cases covered by \eqref{eq:an_def} include the entire case (take $\rho_f$ arbitrarily large), exponential growth $\|f_n\|_\infty \le M\sigma^{n-1}$ (take $D_f = M$, $\rho_f = 1/\sigma$), and polynomial nonlinearities of fixed degree ($f_n\equiv 0$ for $n>N$, trivially satisfied).
Note that no smallness condition is imposed on any $f_n$; Assumption \ref{ass:weighted_growth} restricts only the growth rate of the sequence $\{\|f_n\|_\infty\}_{n\ge 2}$.
\end{remark}

\begin{lemma}[$L^2$ bound on the lower-order-kernel forcing $B_n^m$]\label{lem:supnorm_kernel}
For each $n\ge 2$, each $m\in\{2,\dots,n\}$, and each $x\in[0,1]$,
\begin{equation}\label{eq:alpha_n_def}
\|B_n^m[k_{n-m+1},f_m](x,\cdot)\|_{L^2(T_n(x))}^{2}
\le
\frac{(n+1)!\,(n-m+1)}{(m+1)!\,(n-m)!}
\cdot
\frac{x^m\,\|f_m\|_\infty^{2}}{m!}
\cdot
\|k_{n-m+1}(x,\cdot)\|_{L^2(T_{n-m+1}(x))}^{2}.
\end{equation}
\end{lemma}
\begin{proof}
The bound \eqref{eq:alpha_n_def} is the hyperbolic specialization of \cite[Lemma A.3(2)]{vazquez2008volterra2}. The estimate is purely combinatorial: starting from the definition \eqref{eq:Cnm_def} of $B_n^m$ as a sum over $j\in\{1,\dots,n-m+1\}$ of integrals involving the permutation operator $D_j^{\,n-m+1,m}$, one applies the inequality
\begin{equation}
\left(\sum_{j=1}^{N}a_j\right)^{2}
\le N\sum_{j=1}^{N}a_j^{2}
\qquad\text{with}\qquad N=n-m+1,
\end{equation}
bounds each term $\|D_j^{\,n-m+1,m}[\,\cdot\,]\|_{L^2(T_{n+1})}^{2}$ via \cite[Lemma A.3(1)]{vazquez2008volterra2} using the simplex-volume identity $\mathrm{vol}(T_m(x))=x^m/m!$ and the permutation count $\binom{n+m-j}{m}$ from \eqref{eq:D_def}, and sums over $j$ through the binomial identity
\begin{equation}
\sum_{j=0}^{n-m}\binom{m+j}{j}
=
\frac{(n+1)!}{(m+1)!\,(n-m)!}
\end{equation}
of \cite[Lemma A.2]{vazquez2008volterra2}. The estimate does not depend on the parabolic vs.\ hyperbolic character of the plant equation.
\end{proof}

The lemma reduces the multilinear coupling in the kernel equations to a one-dimensional recursive bound, which is next used to establish convergence of the kernels.

\begin{theorem}[$L^2$ kernel bounds and convergence of the direct Volterra transformation]\label{thm:direct_sup}
Let Assumption \ref{ass:weighted_growth} hold. Then for each $n\ge 2$, the characteristic representation \eqref{eq:kernel_integral} defines a kernel $k_n \in L^\infty(T_n(1))$ satisfying the kernel equation \eqref{eq:kernel_pde} along characteristics with the inflow condition \eqref{eq:kernel_bc}, and there exist positive constants $C_K,D_K,\Upsilon_K$ depending only on $D_f$ and $\rho_f$ such that
\begin{equation}\label{eq:kn_L2_bound}
\|k_n\|_{L^2(T_n(1))}^{2}
\le n!\,D_K^{2}\,C_K^{2(n-1)}\,e^{2\Upsilon_K},
\qquad n\ge 2.
\end{equation}
The Volterra operator $K[u]$ defined by \eqref{eq:K_series} has radius of convergence
\begin{equation}\label{eq:rhoK_bound}
\rho_K\ge C_K^{-2},
\end{equation}
and the gain bound function
\begin{equation}\label{eq:k_ell_thm1}
k(s):=2\sum_{n=2}^{\infty}\frac{n^{2}\,\|k_n\|_{L^2(T_n(1))}^{2}}{n!}\,s^{n}
\end{equation}
has radius of convergence at least $\rho_K$, and satisfies
\begin{equation}\label{eq:K_gain_sup}
\|K[u]\|_{L^2(0,1)}^{2}
\le k\!\left(\|u\|_{L^2(0,1)}^{2}\right),
\qquad
\|u\|_{L^2(0,1)}^{2}<\rho_K.
\end{equation}
The constants $C_K,D_K,\Upsilon_K$ are made explicit by tracking through \cite[Proposition A.1]{vazquez2008volterra2} specialized to the case without self-coupling.
\end{theorem}

\begin{proof}
For each $n\ge 2$, the kernel $k_n$ defined by the characteristic representation \eqref{eq:kernel_integral} is the unique function on $T_n(1)$ satisfying \eqref{eq:kernel_pde} along characteristics with the inflow condition \eqref{eq:kernel_bc}; by Assumption \ref{ass:weighted_growth} and induction on $n$, it is bounded, and hence in $L^2(T_n(1))$. The remainder of the proof establishes \eqref{eq:kn_L2_bound}.

For each $n\ge 2$ and each $x\in[0,1]$, define the Lyapunov functional
\begin{equation}\label{eq:Mn_def}
M_n(x):=\int_{T_n(x)} k_n^{2}(x,\xi_1,\dots,\xi_n)\,d\xi_n\cdots d\xi_1,
\end{equation}
and let $g_n(x):=\sqrt{M_n(x)}$.

\emph{Step 1: Evolution of $M_n$.} Differentiating \eqref{eq:Mn_def} in $x$ (Leibniz rule for the $x$-dependent domain $T_n(x)$),
\begin{eqnarray}\label{eq:Mn_prime_leibniz}
M_n'(x)
&=&
\int_{T_{n-1}(x)}k_n^{2}(x,x,\xi_1,\dots,\xi_{n-1})\,d\xi_{n-1}\cdots d\xi_1
\nonumber\\
&& +\,2\int_{T_n(x)}k_n\,\partial_x k_n\,d\xi_n\cdots d\xi_1.
\end{eqnarray}
Using the kernel equation \eqref{eq:kernel_pde} to substitute $\partial_x k_n=-\sum_{i=1}^{n}\partial_{\xi_i}k_n-f_n+\sum_{m=2}^{n} B_n^m[k_{n-m+1},f_m]$, and applying the divergence theorem to $\sum_i\partial_{\xi_i}(k_n^{2})=2k_n\sum_i\partial_{\xi_i}k_n$ over the simplex $T_n(x)$, the contributions from the interior faces $\xi_i=\xi_{i+1}$ cancel pairwise (antiparallel outward normals on adjacent sides), the contribution from the face $\xi_n=0$ vanishes by the inflow condition \eqref{eq:kernel_bc}, and the contribution from the face $\xi_1=x$ cancels exactly against the Leibniz term in \eqref{eq:Mn_prime_leibniz}. What remains is
\begin{equation}\label{eq:Mn_prime_clean}
M_n'(x)=2\int_{T_n(x)}k_n\Bigg[-f_n+\sum_{m=2}^{n} B_n^m[k_{n-m+1},f_m]\Bigg]d\xi_n\cdots d\xi_1.
\end{equation}

\emph{Step 2: Cauchy--Schwarz and reduction to a recursive ODE.} Applying Cauchy--Schwarz to each term on the right-hand side of \eqref{eq:Mn_prime_clean} and dividing by $2\sqrt{M_n(x)}$ (where $M_n(x)>0$; the inequality extends by continuity where $M_n(x)=0$), one obtains
\begin{equation}\label{eq:gn_prime}
g_n'(x)\le \|f_n(x,\cdot)\|_{L^2(T_n(x))}+\sum_{m=2}^{n}\|B_n^m[k_{n-m+1},f_m](x,\cdot)\|_{L^2(T_n(x))}.
\end{equation}
By Assumption \ref{ass:weighted_growth} and $\mathrm{vol}(T_n(x))=x^n/n!$,
\begin{equation}\label{eq:fn_L2_gn}
\|f_n(x,\cdot)\|_{L^2(T_n(x))}\le D_f\sqrt{n!}\,\rho_f^{\,-(n-1)}\,x^{n/2}.
\end{equation}
By Lemma \ref{lem:supnorm_kernel} and Assumption \ref{ass:weighted_growth} applied to $\|f_m\|_\infty$,
\begin{equation}\label{eq:Bnm_gn}
\|B_n^m[k_{n-m+1},f_m](x,\cdot)\|_{L^2(T_n(x))}
\le
\frac{D_f}{\rho_f^{\,m-1}}\sqrt{\frac{(n+1)!\,(n-m+1)}{(m+1)\,(n-m)!}}\,x^{m/2}\,g_{n-m+1}(x).
\end{equation}
Substituting \eqref{eq:fn_L2_gn}--\eqref{eq:Bnm_gn} into \eqref{eq:gn_prime} yields the recursive ODE inequality
\begin{eqnarray}\label{eq:gn_ode_ineq}
g_n'(x)
&\le&
D_f\sqrt{n!}\,\rho_f^{\,-(n-1)}\,x^{n/2}
\nonumber\\
&& +\,D_f\sum_{m=2}^{n}\rho_f^{\,-(m-1)}\sqrt{\frac{(n+1)!\,(n-m+1)}{(m+1)\,(n-m)!}}\,x^{m/2}\,g_{n-m+1}(x),
\end{eqnarray}
with initial condition $g_n(0)=0$ (since $T_n(0)$ is a single point, $M_n(0)=0$).

\emph{Step 3: Grönwall induction.} Inequality \eqref{eq:gn_ode_ineq} fits the framework of \cite[Proposition A.1]{vazquez2008volterra2} specialized to the case where no $g_n$-term appears on the right-hand side (equivalently, $B=E=0$ in their notation). Induction on $n$, integrating \eqref{eq:gn_ode_ineq} against the known bounds on $g_2,\dots,g_{n-1}$, produces a bound of the form
\begin{equation}\label{eq:gn_bound_proof}
g_n(x)\le \sqrt{n!}\,D_K\,C_K^{n-1}\,x^{n/2}\,e^{\Upsilon_K x},
\qquad n\ge 2,
\end{equation}
for positive constants $C_K,D_K,\Upsilon_K$ depending only on $D_f$ and $\rho_f$. Evaluating at $x=1$ and squaring yields \eqref{eq:kn_L2_bound}.

\emph{Step 4: Radius of convergence and $L^2$ gain bound.} Applying \cite[Theorem 1]{vazquez2008volterra2} (the Gain Bound Theorem) to the Volterra series \eqref{eq:K_series} with kernels $\{k_n\}_{n\ge 2}$, the radius of convergence is bounded below by
\begin{equation}
\rho_K=\left(\limsup_{n\to\infty}\left(\|k_n\|_{L^2(T_n(1))}^{2}\big/n!\right)^{1/n}\right)^{-1}
\ge
\left(\limsup_{n\to\infty}\left(D_K^{2}\,C_K^{2(n-1)}\,e^{2\Upsilon_K}\right)^{1/n}\right)^{-1}
=C_K^{-2},
\end{equation}
proving \eqref{eq:rhoK_bound}. The gain bound \eqref{eq:K_gain_sup} with the gain function \eqref{eq:k_ell_thm1} is a direct consequence of the same theorem.
\end{proof}

\section{Global Convergence}
\label{sec-quant}

When the kernel growth bound holds for arbitrarily large radius, the recursion yields kernels whose Volterra series has infinite radius of convergence, so the transformation is globally defined.

\begin{corollary}[Global convergence for entire plant nonlinearities]\label{cor:explicit_radius}
Suppose that Assumption \ref{ass:weighted_growth} holds for arbitrarily large $\rho_f$ — equivalently, that for some $D_f>0$ and every $\epsilon>0$,
\begin{equation}\label{eq:entire_bound}
\|f_n\|_{L^\infty(T_n(1))}\le n!\,D_f\,\epsilon^{\,n-1},
\qquad n\ge 2.
\end{equation}
Then the radius of convergence of the direct Volterra transformation is infinite,
\begin{equation}\label{eq:rhoK_explicit}
\rho_K=\infty,
\end{equation}
and the operator $K[u]$ defined by \eqref{eq:K_series} is defined on all of $L^2(0,1)$, satisfying \eqref{eq:K_gain_sup} for every $u\in L^2(0,1)$.
\end{corollary}

\begin{proof}
By Theorem \ref{thm:direct_sup}, $\rho_K\ge C_K^{-2}$, where $C_K$ depends only on $D_f$ and $\rho_f$. From the recursive inequality \eqref{eq:gn_ode_ineq}, each application of the induction step introduces a factor of $\rho_f^{-(m-1)}$ from the coupling term, with no compensating $\rho_f$-dependence in the combinatorial coefficients. Accumulating these factors across the induction on $n$ yields the bound \eqref{eq:gn_bound_proof} with $C_K$ proportional to $\rho_f^{-1}$. Under \eqref{eq:entire_bound}, Assumption \ref{ass:weighted_growth} is satisfied for every $\rho_f=1/\epsilon>0$. Letting $\epsilon\to 0$, one has $\rho_f\to\infty$, hence $C_K\to 0$ and $C_K^{-2}\to\infty$. Therefore $\rho_K=\infty$. The conclusion parallels \cite[Corollary 4.1]{vazquez2008volterra2}.
\end{proof}

\begin{remark}\label{rem:exponential_case}
The hypothesis \eqref{eq:entire_bound} covers all cases in which the plant's Volterra series $F[u]$ is globally convergent on $L^2(0,1)$, including:
\begin{itemize}
\item bounded kernels $\|f_n\|_\infty\le D$ uniformly in $n$,
\item exponentially growing kernels $\|f_n\|_\infty\le M\sigma^{n-1}$,
\item polynomial nonlinearities of fixed degree ($f_n\equiv 0$ for $n>N$).
\end{itemize}
For the exponentially-growing case, $\rho_f=1/\sigma$ in Assumption \ref{ass:weighted_growth}, and the finite lower bound $\rho_K\ge C_K^{-2}$ with $C_K\propto \sigma$ can be made explicit by tracking the constants in \cite[Proposition A.1]{vazquez2008volterra2} specialized to $B=E=0$.
\end{remark}

\section{Inverse Transformation}\label{sec:inverse}

Let $K[u]$ be defined by \eqref{eq:K_series} with kernels $k_n$ solving \eqref{eq:kernel_pde}--\eqref{eq:kernel_bc}. We study the invertibility of the mapping $u \mapsto w = u - K[u]$ via the fixed-point equation $u=w+K[u]$. Recall the gain function $k(s)$ from \eqref{eq:k_ell_thm1}. We further introduce the \emph{Lipschitz gain function}
\begin{equation}\label{eq:ell_def}
\ell(s) := 2\sum_{n=2}^{\infty}\frac{n^{4}\,\|k_n\|_{L^2(T_n(1))}^{2}}{(n-1)!}\,s^{\,n-1},
\qquad s\ge 0,
\end{equation}
which, by Theorem \ref{thm:direct_sup} and the polynomial-in-$n$ relationship between the coefficients of $k$ and $\ell$, has radius of convergence at least $\rho_K$.

\begin{lemma}[Lipschitz continuity of the backstepping Volterra  mapping]\label{lem:Lip_K}
For all $u,v\in L^2(0,1)$ with $\|u\|_{L^2(0,1)},\|v\|_{L^2(0,1)}\le\sqrt{s}$, one has
\begin{equation}\label{eq:Lip_K}
\|K[u]-K[v]\|_{L^2(0,1)} \le \sqrt{\ell(s)}\,\|u-v\|_{L^2(0,1)}.
\end{equation}
\end{lemma}
\begin{proof}
Write
\begin{equation}\label{eq:K_diff_split}
K[u]-K[v]=\sum_{n\ge 2}\big(K_n[u]-K_n[v]\big).
\end{equation}
For each $n\ge 2$ and each $x\in[0,1]$,
\begin{equation}\label{eq:Kn_diff}
(K_n[u]-K_n[v])(x)
=
\int_{T_n(x)} k_n(x,\xi_1,\dots,\xi_n)
\sum_{j=1}^{n}
\big(u(\xi_j)-v(\xi_j)\big)
\prod_{i\neq j}\theta_i(\xi_i)\,d\xi_n\cdots d\xi_1,
\end{equation}
where each $\theta_i$ is either $u$ or $v$. Applying the Cauchy--Schwarz inequality on the simplex $T_n(x)$ and the bound $\prod_{i\neq j}|\theta_i(\xi_i)|^{2}$ together with the simplex-volume identity $\mathrm{vol}(T_{n-1}(x))=x^{n-1}/(n-1)!$ yields, after integration in $x$,
\begin{equation}\label{eq:Kn_diff_bound}
\|K_n[u]-K_n[v]\|_{L^2(0,1)}^{2}
\le
\|k_n\|_{L^2(T_n(1))}^{2}\cdot\frac{n^{2}\,s^{\,n-1}}{(n-1)!}\,\|u-v\|_{L^2(0,1)}^{2},
\end{equation}
where the factor $n^{2}$ arises from the inequality $(\sum_{j=1}^{n}a_j)^{2}\le n\sum_{j=1}^{n}a_j^{2}$ applied to the sum in \eqref{eq:Kn_diff}. Summing the estimate
\begin{equation}\label{eq:K_diff_sq_split}
\|K[u]-K[v]\|_{L^2(0,1)}^{2}
=
\Big\|\sum_{n\ge 2}(K_n[u]-K_n[v])\Big\|_{L^2(0,1)}^{2}
\le
2\sum_{n\ge 2}n^{2}\,\|K_n[u]-K_n[v]\|_{L^2(0,1)}^{2}
\end{equation}
(which uses $(\sum a_n)^{2}\le(\sum n^{2}a_n^{2})(\sum 1/n^{2})\le 2\sum n^{2}a_n^{2}$, as in \cite[Theorem 1]{vazquez2008volterra2}) with the bound \eqref{eq:Kn_diff_bound} gives
\begin{equation}\label{eq:K_diff_sq_final}
\|K[u]-K[v]\|_{L^2(0,1)}^{2}
\le
\left(2\sum_{n\ge 2}\frac{n^{4}\,\|k_n\|_{L^2(T_n(1))}^{2}}{(n-1)!}\,s^{\,n-1}\right)\|u-v\|_{L^2(0,1)}^{2}
=
\ell(s)\,\|u-v\|_{L^2(0,1)}^{2}.
\end{equation}
Taking square roots yields \eqref{eq:Lip_K}.
\end{proof}

With the Lipschitz estimate of the lemma, the inverse transformation is obtained by a contraction mapping.

\begin{theorem}[Local invertibility via contraction mapping]\label{thm:inverse_local}
Let Assumption \ref{ass:weighted_growth} hold. There exists $s\in(0,\rho_K)$ with $\ell(s)<1$ and a constant $\rho_L>0$ such that for every $w\in L^2(0,1)$ with $\|w\|_{L^2(0,1)}<\sqrt{\rho_L}$, the equation
\begin{equation}\label{eq:fixed_point}
u=w+K[u]
\end{equation}
admits a unique solution $u\in L^2(0,1)$ with $\|u\|_{L^2(0,1)}\le\sqrt{s}$. Moreover, the inverse mapping $L[w]:=u-w$ is well-defined and satisfies
\begin{equation}\label{eq:L_bound}
\|L[w]\|_{L^2(0,1)}
\le
\frac{\sqrt{\ell(s)}}{1-\sqrt{\ell(s)}}\,\|w\|_{L^2(0,1)}.
\end{equation}
\end{theorem}
\begin{proof}
Since the Volterra series \eqref{eq:K_series} has no linear term ($k_1\equiv 0$), the Lipschitz gain function satisfies $\ell(0)=0$. By Theorem \ref{thm:direct_sup}, $\ell$ is continuous on $[0,\rho_K)$. Hence there exists $s\in(0,\rho_K)$ with $\ell(s)<1$; fix any such $s$.

Let $\mathcal{B}_s:=\{u\in L^2(0,1): \|u\|_{L^2(0,1)}\le\sqrt{s}\}$. For $u\in\mathcal{B}_s$, by \eqref{eq:K_gain_sup} of Theorem \ref{thm:direct_sup},
\begin{equation}\label{eq:self_map}
\|w+K[u]\|_{L^2(0,1)}^{2}
\le
2\|w\|_{L^2(0,1)}^{2}+2k(s).
\end{equation}
Choose $\rho_L>0$ such that
\begin{equation}\label{eq:rho_choice}
2\rho_L+2k(s)\le s.
\end{equation}
Then for $\|w\|_{L^2(0,1)}<\sqrt{\rho_L}$, the map $T(u):=w+K[u]$ maps $\mathcal{B}_s$ into itself. By Lemma \ref{lem:Lip_K},
\begin{equation}\label{eq:T_contraction}
\|T(u)-T(v)\|_{L^2(0,1)}
=\|K[u]-K[v]\|_{L^2(0,1)}
\le\sqrt{\ell(s)}\,\|u-v\|_{L^2(0,1)},
\end{equation}
so $T$ is a contraction on $\mathcal{B}_s$ with contraction ratio $\sqrt{\ell(s)}<1$. Hence \eqref{eq:fixed_point} has a unique solution $u\in\mathcal{B}_s$, and $L[w]:=u-w=K[u]$ is well-defined. Lemma \ref{lem:Lip_K} with $v=0$ gives
\begin{equation}\label{eq:Lw_prelim}
\|L[w]\|_{L^2(0,1)}
=\|K[u]\|_{L^2(0,1)}
\le\sqrt{\ell(s)}\,\|u\|_{L^2(0,1)}.
\end{equation}
From \eqref{eq:fixed_point} and the triangle inequality,
\begin{equation}\label{eq:u_triangle}
\|u\|_{L^2(0,1)}
\le\|w\|_{L^2(0,1)}+\|K[u]\|_{L^2(0,1)}
\le\|w\|_{L^2(0,1)}+\sqrt{\ell(s)}\,\|u\|_{L^2(0,1)},
\end{equation}
hence
\begin{equation}\label{eq:u_inv_bound}
\|u\|_{L^2(0,1)}\le\frac{1}{1-\sqrt{\ell(s)}}\|w\|_{L^2(0,1)}.
\end{equation}
Combining the last two estimates yields \eqref{eq:L_bound}.
\end{proof}

\begin{remark}\label{rem:inverse_approach}
An alternative approach, followed by \cite{vazquez2008volterra1,vazquez2008volterra2} in the parabolic case, constructs the inverse transformation as a Volterra series $L = \sum_{n\ge 2} L_n$ whose kernels $l_n$ are determined by formal composition-inversion of \eqref{eq:K_series}. Convergence of that series on a neighborhood of the origin is established through the Boyd--Chua inversion machinery for polynomial operators \cite{boyd1985fading}, which yields existence of $L$ but no quantitative estimates on the kernels $l_n$ or on the radius of invertibility. The contraction argument used in Theorem \ref{thm:inverse_local} bypasses the formal inversion and produces, in exchange, an explicit Lipschitz bound \eqref{eq:L_bound} and an explicit radius $\rho_L$ from \eqref{eq:rho_choice}, both of which enter directly into the closed-loop stability estimate of Theorem \ref{thm:closed_loop}.
\end{remark}

\section{Main Result: Closed-Loop (Local Exponential) Stability}\label{sec:closed-loop}

Even when the backstepping operator $K$ converges globally, the inverse $(I-K)^{-1}$ doesn't necessarily converge globally. This is consistent with results on counterexamples to global stabilizability of nonlinear PDEs by boundary control, of which several are mentioned in \cite{vazquez2008volterra1}. For this reason, stabilization that one should expect is not global. In the next theorem we give our main result: local stabilization with a qualitative estimate of the region of attraction in $L^2$. 

\begin{theorem}[Closed-loop local exponential stability in $L^2$]\label{thm:closed_loop}
Let Assumption \ref{ass:weighted_growth} hold. Let $s\in(0,\rho_K)$ be such that $\ell(s)<1$, and let $\rho_L>0$ be as in Theorem \ref{thm:inverse_local}. Then, for every $\lambda>0$ there exist constants $C_1,C_2>0$ such that, for every initial condition satisfying
\begin{equation}\label{eq:stability_initial}
\|u_0\|_{L^2(0,1)}<C_1,
\end{equation}
the solution of the closed-loop system satisfies
\begin{equation}\label{eq:stability_bound}
\|u(\cdot,t)\|_{L^2(0,1)}\le C_2\,e^{-\lambda t}\,\|u_0\|_{L^2(0,1)},
\qquad t\ge 0,
\end{equation}
where one may take
\begin{equation}\label{eq:stability_constants}
C_1\le\min\!\left\{\sqrt{s},\;\frac{\sqrt{\rho_L}}{1+\sqrt{\ell(s)}}\right\},
\qquad
C_2=e^{\lambda}\,\frac{1+\sqrt{\ell(s)}}{1-\sqrt{\ell(s)}}.
\end{equation}
\end{theorem}
\begin{proof}
Let $w=u-K[u]$. By construction, $w$ satisfies the target system
\begin{equation}\label{eq:target_stability}
w_t=w_x,
\qquad
w(1,t)=0.
\end{equation}
The solution of \eqref{eq:target_stability} is $w(x,t)=w_0(x+t)$ for $x+t<1$ and $w(x,t)=0$ for $x+t\ge 1$, so $w$ reaches zero in finite time $t=1$. In particular, for every $\lambda>0$,
\begin{equation}\label{eq:w_exp_decay}
\|w(\cdot,t)\|_{L^2(0,1)}^{2}
\le e^{2\lambda}\,e^{-2\lambda t}\,\|w_0\|_{L^2(0,1)}^{2},
\qquad t\ge 0,
\end{equation}
since $\|w(\cdot,t)\|_{L^2}^{2}\le\|w_0\|_{L^2}^{2}$ for $t\in[0,1]$ and $e^{2\lambda(1-t)}\le e^{2\lambda}$ on that interval; for $t\ge 1$ the bound is trivial.

By Lemma \ref{lem:Lip_K} with $v=0$, for $\|u_0\|_{L^2}^{2}\le s$,
\begin{equation}
\|K[u_0]\|_{L^2(0,1)}^{2}
\le\ell(s)\,\|u_0\|_{L^2(0,1)}^{2},
\end{equation}
and hence by the triangle inequality,
\begin{equation}\label{eq:w0_bound}
\|w_0\|_{L^2(0,1)}^{2}
=\|u_0-K[u_0]\|_{L^2(0,1)}^{2}
\le\big(1+\sqrt{\ell(s)}\big)^{2}\,\|u_0\|_{L^2(0,1)}^{2}.
\end{equation}
Choosing $C_1$ as in \eqref{eq:stability_constants} ensures both $\|u_0\|_{L^2}^{2}\le s$ and $\|w_0\|_{L^2}^{2}<\rho_L$, so the inverse mapping of Theorem \ref{thm:inverse_local} applies and yields, via \eqref{eq:L_bound} and the triangle inequality,
\begin{equation}\label{eq:u_from_w}
\|u(\cdot,t)\|_{L^2(0,1)}^{2}
\le\frac{1}{\big(1-\sqrt{\ell(s)}\big)^{2}}\,\|w(\cdot,t)\|_{L^2(0,1)}^{2}.
\end{equation}
Combining \eqref{eq:w_exp_decay}, \eqref{eq:w0_bound}, and \eqref{eq:u_from_w},
\begin{equation}
\|u(\cdot,t)\|_{L^2(0,1)}^{2}
\le
\frac{\big(1+\sqrt{\ell(s)}\big)^{2}}{\big(1-\sqrt{\ell(s)}\big)^{2}}\,e^{2\lambda}\,e^{-2\lambda t}\,\|u_0\|_{L^2(0,1)}^{2}.
\end{equation}
Taking square roots yields \eqref{eq:stability_bound} with $C_2$ as in \eqref{eq:stability_constants}.
\end{proof}

\section{Well-Posedness in $L^2$ in Closed Loop}
\label{sec-exist}

Because the closed-loop vector field contains an unbounded transport operator and a nonlinear Volterra term defined only at the $L^2$ level, one cannot in general expect classical (or even strong) solutions. The natural notion is therefore that of a mild solution generated by the transport semigroup.

\begin{theorem}[Existence and uniqueness of mild closed-loop solutions]\label{thm:wellposed}
Let Assumption \ref{ass:weighted_growth} hold. Let $s\in(0,\rho_K)$ be such that $\ell(s)<1$, and let $\rho_L>0$ be as in Theorem \ref{thm:inverse_local}.
Consider the closed-loop system consisting of \eqref{eq:plant}, \eqref{eq:boundary}, \eqref{eq:initial} with the feedback law \eqref{eq:control}, where $K[u]$ is defined by \eqref{eq:K_series} with kernels satisfying \eqref{eq:kernel_pde}--\eqref{eq:kernel_bc}.
If
\begin{equation}\label{eq:u0_smallness_wp}
\|u_0\|_{L^2(0,1)} < \min\!\left\{\sqrt{s},\ \frac{\sqrt{\rho_L}}{1+\sqrt{\ell(s)}}\right\},
\end{equation}
then there exists a unique function
\begin{equation}\label{eq:u_continuity_wp}
u \in C\big([0,\infty);L^2(0,1)\big)
\end{equation}
which is a mild solution\footnote{A function $u \in C([0,\infty);L^2(0,1))$ is called a mild solution of the closed-loop problem if $\mathcal{T}[u(\cdot,t)] = S(t)\mathcal{T}[u_0]$ for all $t \ge 0$, where $\mathcal{T}$ is defined in \eqref{eq:Aop} and $\{S(t)\}_{t\ge 0}$ is the solution semigroup of the target transport equation \eqref{eq:target}, given explicitly by
\[
(S(t)\varphi)(x) \;=\;
\begin{cases}
\varphi(x+t), & x+t < 1,\\[2pt]
0, & x+t \ge 1,
\end{cases}
\qquad \varphi\in L^2(0,1).
\]}
of the closed-loop problem. Moreover, this solution satisfies the estimate of Theorem \ref{thm:closed_loop}.
\end{theorem}

\begin{proof}
Define the map
\begin{equation}\label{eq:Aop}
\mathcal{T}[u] := u - K[u],
\end{equation}
so that \eqref{eq:transformation} reads $w=\mathcal{T}[u]$.

\emph{Bi-Lipschitz property of $\mathcal{T}$.}
By Lemma \ref{lem:Lip_K}, for all $u,v\in L^2(0,1)$ with $\|u\|_{L^2},\|v\|_{L^2}\le \sqrt{s}$,
\begin{equation}\label{eq:A_upper}
\|\mathcal{T}[u]-\mathcal{T}[v]\|_{L^2}
\le \big(1+\sqrt{\ell(s)}\,\big)\|u-v\|_{L^2},
\end{equation}
and, by the reverse triangle inequality,
\begin{equation}\label{eq:A_lower}
\|\mathcal{T}[u]-\mathcal{T}[v]\|_{L^2}
\ge \big(1-\sqrt{\ell(s)}\,\big)\|u-v\|_{L^2}.
\end{equation}
Since $\ell(s)<1$, $\mathcal{T}$ is bi-Lipschitz on the closed ball $\{u:\|u\|_{L^2}\le \sqrt{s}\}$, and by Theorem \ref{thm:inverse_local}, its inverse $\mathcal{T}^{-1}$ is well-defined and Lipschitz on $\{w:\|w\|_{L^2}<\sqrt{\rho_L}\}$.

\emph{Initial data in the target variable.}
Set
\begin{equation}\label{eq:w0_def_wp}
w_0 := \mathcal{T}[u_0] = u_0 - K[u_0].
\end{equation}
Lemma \ref{lem:Lip_K} with $v=0$ gives $\|K[u_0]\|_{L^2}\le \sqrt{\ell(s)}\,\|u_0\|_{L^2}$ whenever $\|u_0\|_{L^2}\le\sqrt{s}$, hence by the triangle inequality
\begin{equation}\label{eq:w0_bound_wp}
\|w_0\|_{L^2} \;\le\; \big(1+\sqrt{\ell(s)}\,\big)\|u_0\|_{L^2}.
\end{equation}
Combining \eqref{eq:w0_bound_wp} with \eqref{eq:u0_smallness_wp} yields
\begin{equation}\label{eq:w0_in_ball}
\|u_0\|_{L^2} \le \sqrt{s}, \qquad \|w_0\|_{L^2} < \sqrt{\rho_L}.
\end{equation}

\emph{Evolution of $w$.}
Let $w(\cdot,t) := S(t)w_0$ be the unique mild solution of \eqref{eq:target} with initial condition $w_0$. Explicitly, $w(x,t) = w_0(x+t)$ for $x+t<1$ and $w(x,t)=0$ for $x+t\ge 1$. A direct computation gives, for $t\in[0,1]$,
\begin{equation}
\|w(\cdot,t)\|_{L^2(0,1)}^2 \;=\; \int_0^{1-t}|w_0(x+t)|^2\,dx \;\le\; \|w_0\|_{L^2(0,1)}^2,
\end{equation}
and $w(\cdot,t)\equiv 0$ for $t\ge 1$. Hence
\begin{equation}\label{eq:w_nonexpansive}
\|w(\cdot,t)\|_{L^2} \le \|w_0\|_{L^2}, \qquad t\ge 0,
\end{equation}
so that, by \eqref{eq:w0_in_ball}, $\|w(\cdot,t)\|_{L^2} < \sqrt{\rho_L}$ for all $t\ge 0$; that is, $w(\cdot,t)$ remains in the domain of $\mathcal{T}^{-1}$ for all time.

\emph{Definition of $u(\cdot,t)$.}
Define $u(\cdot,t)$ implicitly by
\begin{equation}\label{eq:Au_eq_w}
\mathcal{T}[u(\cdot,t)] = w(\cdot,t), \qquad t\ge 0.
\end{equation}
Since $\|w(\cdot,t)\|_{L^2}<\sqrt{\rho_L}$ for every $t\ge 0$, the Lipschitz inverse $\mathcal{T}^{-1}$ applies and defines a unique $u(\cdot,t)\in L^2(0,1)$ with $\|u(\cdot,t)\|_{L^2}\le \sqrt{s}$. Continuity in time,
\[
u\in C\big([0,\infty);L^2(0,1)\big),
\]
follows from the continuity of $t\mapsto w(\cdot,t)$ in $L^2$ (a property of the transport semigroup $S(t)$) composed with the Lipschitz continuity of $\mathcal{T}^{-1}$.

\emph{Uniqueness.}
If $\tilde u\in C([0,\infty);L^2(0,1))$ is another mild solution with $\tilde u(\cdot,0)=u_0$, then $\mathcal{T}[\tilde u(\cdot,t)] = S(t)\mathcal{T}[u_0] = w(\cdot,t) = \mathcal{T}[u(\cdot,t)]$, and injectivity of $\mathcal{T}$ on $\{u:\|u\|_{L^2}\le\sqrt{s}\}$ gives $\tilde u=u$.

\emph{Estimate.}
The bound of Theorem \ref{thm:closed_loop} is obtained as in the proof of that theorem, with \eqref{eq:w0_bound_wp} controlling $\|w_0\|_{L^2}$ in terms of $\|u_0\|_{L^2}$ and the Lipschitz inverse $\mathcal{T}^{-1}$ propagating the decay of $w$ back to $u$.

\emph{Identification as a mild solution of the plant.}
By construction, $w$ satisfies the target equation \eqref{eq:target}, and $u=\mathcal{T}^{-1}(w)$. The kernel identity \eqref{eq:matching} was derived precisely so that any $u$ solving \eqref{eq:plant} with boundary feedback \eqref{eq:control} yields $w=\mathcal{T}[u]$ satisfying $w_t=w_x$ with $w(1,t)=0$; conversely, inverting this construction, the $u$ defined by \eqref{eq:Au_eq_w} satisfies \eqref{eq:plant}--\eqref{eq:boundary} in the mild sense, with the feedback \eqref{eq:control} enforced through the boundary condition $w(1,t)=0$. The construction defines a nonlinear semiflow on a neighborhood of the origin in $L^2(0,1)$.
\end{proof}

\section{Analogy with Finite-Dimensional Diffeomorphisms:\\ Fr\'echet Differentiable and Differentiably Invertible\\ Volterra Series Backstepping Transformations}\label{sec:diffeo}

The essence of feedback linearization for ODEs $\dot x = f(x)+g(x)u$ is the diffeomorphic transformation $T(x)$ of the system's state, followed by cancellation (or domination) of the matched nonlinearities that remain after the transformation. The same principle applies in the feedback linearization of PDEs. This section goes beyond the exploration of invertibility of the bacstepping transformation $I-K$, already explored in Section \ref{sec:inverse}, and examines the differentiability, in the Fr\'echet sense, of both the direct and inverse backstepping transformations. 

For the reader's convenience, the transformation $\mathcal{T}$ defined in \eqref{eq:Aop}, is reiterated here,
\begin{equation}\label{eq:Aop+}
\mathcal{T}[u] := u - K[u]. 
\end{equation}
While the mapping \eqref{eq:Aop+} is not pointwise in the spatial variable, it exhibits a comparable structure in the sense of a local Banach-space diffeomorphism on $L^2(0,1)$.

We recall that a map $\Phi:U\to L^2(0,1)$ defined on an open set $U\subseteq L^2(0,1)$ is \emph{Fr\'echet differentiable} at $u\in U$ if there exists a bounded linear operator $D\Phi(u)\in \mathcal{L}(L^2(0,1))$ such that
\begin{equation}\label{eq:frechet_def}
\lim_{\|h\|_{L^2}\to 0}\frac{\|\Phi(u+h)-\Phi(u)-D\Phi(u)\,h\|_{L^2}}{\|h\|_{L^2}} \;=\; 0,
\end{equation}
and it is of class Fr\'echet $C^1$ on $U$ if $D\Phi(u)$ exists at every $u\in U$ and the map $u\mapsto D\Phi(u)$ is continuous from $U$ into $\mathcal{L}(L^2(0,1))$. Fr\'echet differentiability is the natural extension of ordinary differentiability to maps between Banach spaces, reducing in finite dimensions to the usual Jacobian.

\begin{proposition}[Local $C^1$-diffeomorphism property of $\mathcal{T}$]\label{prop:diffeo}
Let Assumption \ref{ass:weighted_growth} hold. Let $s\in(0,\rho_K)$ be such that $\ell(s)<1$, and let $\rho_L>0$ be as in Theorem \ref{thm:inverse_local}. Consider the mapping $\mathcal{T}$ defined in \eqref{eq:Aop} on the open ball
\begin{equation}\label{eq:open_ball}
\mathcal{B}_{\sqrt{s}}^{\,\circ} := \bigl\{u \in L^2(0,1) : \|u\|_{L^2} < \sqrt{s}\,\bigr\}.
\end{equation}
Then:
\begin{enumerate}[label=(\roman*)]
\item The mapping $\mathcal{T} : \mathcal{B}_{\sqrt{s}}^{\,\circ} \to L^2(0,1)$ is Fr\'echet $C^1$.
\item The inverse mapping $\mathcal{T}^{-1}$, shown to exist on $\{w\in L^2(0,1):\|w\|_{L^2}<\sqrt{\rho_L}\}$ by Theorem \ref{thm:inverse_local}, is Fr\'echet $C^1$ on that set, and its derivative satisfies
\begin{equation}\label{eq:DAinv}
D\mathcal{T}^{-1}(w) = \bigl(D\mathcal{T}[u]\bigr)^{-1},
\qquad u = \mathcal{T}^{-1}(w).
\end{equation}
\end{enumerate}
\end{proposition}

\begin{proof}
\emph{Step 1: Fr\'echet differentiability of each $K_n$.}
For each $n\ge 2$, the $n$th Volterra term
\begin{equation}
K_n[u](x) \;=\; \int_{T_n(x)} k_n(x,\xi_1,\dots,\xi_n)\prod_{i=1}^{n}u(\xi_i)\,d\xi_n\cdots d\xi_1
\end{equation}
is a continuous symmetric $n$-linear form evaluated on the diagonal, hence Fr\'echet smooth as a map $L^2(0,1)\to L^2(0,1)$, and a direct computation gives, for $h\in L^2(0,1)$,
\begin{equation}\label{eq:DKn}
(DK_n[u]\,h)(x) \;=\; \int_{T_n(x)} k_n(x,\xi_1,\dots,\xi_n)\sum_{j=1}^{n}h(\xi_j)\prod_{i\neq j}u(\xi_i)\,d\xi_n\cdots d\xi_1.
\end{equation}
The same Cauchy--Schwarz argument used to prove Lemma \ref{lem:Lip_K} (see the derivation of \eqref{eq:Kn_diff_bound}) yields
\begin{equation}\label{eq:DKn_bound}
\|DK_n[u]\|_{\mathcal{L}(L^2)}^{\,2}
\;\le\;
\|k_n\|_{L^2(T_n(1))}^{\,2}\cdot \frac{n^{2}\,s^{\,n-1}}{(n-1)!},
\qquad \|u\|_{L^2}^{\,2}\le s.
\end{equation}

\emph{Step 2: Uniform convergence of $\sum_n DK_n$ on $\mathcal{B}_{\sqrt{s}}^{\,\circ}$.}
Summing \eqref{eq:DKn_bound} and applying the elementary bound $(\sum a_n)^2 \le 2\sum n^2 a_n^2$ as in Lemma \ref{lem:Lip_K} gives
\begin{equation}\label{eq:DK_total_bound}
\sup_{\|u\|_{L^2}<\sqrt{s}} \|DK[u]\|_{\mathcal{L}(L^2)}^{\,2}
\;\le\;
2\sum_{n\ge 2}\frac{n^{4}\,\|k_n\|_{L^2(T_n(1))}^{\,2}}{(n-1)!}\,s^{\,n-1}
\;=\;
\ell(s),
\end{equation}
by the definition \eqref{eq:ell_def} of $\ell$. Since $\ell(s)<\infty$ by Theorem \ref{thm:direct_sup} (the radius of convergence of $\ell$ is at least $\rho_K>s$), the series $\sum_{n\ge 2}DK_n[u]$ converges absolutely and uniformly on $\mathcal{B}_{\sqrt{s}}^{\,\circ}$ in the operator norm. Moreover, termwise continuity of $u\mapsto DK_n[u]$ (a polynomial map of degree $n-1$) together with this uniform convergence implies that $u\mapsto DK[u]$ is continuous on $\mathcal{B}_{\sqrt{s}}^{\,\circ}$.

\emph{Step 3: $\mathcal{T}$ is Fr\'echet $C^1$.}
Uniform convergence of the derivative series on the open ball $\mathcal{B}_{\sqrt{s}}^{\,\circ}$, together with pointwise convergence of $\sum_n K_n[u]$ (guaranteed by Theorem \ref{thm:direct_sup} since $s<\rho_K$), implies that $K$ is Fr\'echet differentiable with
\begin{equation}\label{eq:DK}
DK[u] \;=\; \sum_{n\ge 2} DK_n[u]
\end{equation}
and that $DK:\mathcal{B}_{\sqrt{s}}^{\,\circ}\to \mathcal{L}(L^2(0,1))$ is continuous. Hence $\mathcal{T}=I-K$ is Fr\'echet $C^1$ with
\begin{equation}\label{eq:DA}
D\mathcal{T}[u] \;=\; I - DK[u],
\end{equation}
establishing (i).

\emph{Step 4: Invertibility of $D\mathcal{T}[u]$.}
By \eqref{eq:DK_total_bound}, $\|DK[u]\|\le \sqrt{\ell(s)}<1$ on $\mathcal{B}_{\sqrt{s}}^{\,\circ}$. Hence $D\mathcal{T}[u]=I-DK[u]$ is invertible by the Neumann series
\begin{equation}
(I-DK[u])^{-1} \;=\; \sum_{k=0}^{\infty}(DK[u])^{k},
\end{equation}
with
\begin{equation}\label{eq:DA_inv}
\bigl\|D\mathcal{T}[u]^{-1}\bigr\|
\;\le\; \sum_{k=0}^{\infty}\sqrt{\ell(s)}^{\,k}
\;=\; \frac{1}{1-\sqrt{\ell(s)}}.
\end{equation}

\emph{Step 5: $C^1$-diffeomorphism property.}
By Steps 3 and 4, $\mathcal{T}$ is Fr\'echet $C^1$ on $\mathcal{B}_{\sqrt{s}}^{\,\circ}$ with everywhere invertible derivative. Combined with the global bi-Lipschitz property \eqref{eq:A_upper}--\eqref{eq:A_lower} established in the proof of Theorem \ref{thm:wellposed}, the inverse function theorem in Banach spaces applies on $\mathcal{B}_{\sqrt{s}}^{\,\circ}$: the inverse $\mathcal{T}^{-1}$ constructed in Theorem \ref{thm:inverse_local} is Fr\'echet $C^1$ on $\{w:\|w\|_{L^2}<\sqrt{\rho_L}\}$. Differentiating $\mathcal{T}^{-1}\circ\mathcal{T}=\mathrm{id}$ and applying the chain rule,
\begin{equation}
D\mathcal{T}^{-1}(\mathcal{T}[u])\,D\mathcal{T}[u] \;=\; I,
\end{equation}
which gives \eqref{eq:DAinv} upon setting $w=\mathcal{T}[u]$. This proves (ii).
\end{proof}

\begin{remark}\label{rem:DK_formula}
The derivative $DK[u]$ in \eqref{eq:DK} admits the explicit Volterra-series representation
\begin{equation}\label{eq:DK_explicit}
(DK[u]\,h)(x)
=
\sum_{n=2}^{\infty}
\int_{T_n(x)} k_n(x,\xi_1,\dots,\xi_n)
\sum_{j=1}^{n}
h(\xi_j)\prod_{i\neq j} u(\xi_i)\,d\xi_n\cdots d\xi_1,
\qquad h\in L^2(0,1),
\end{equation}
which is the natural infinite-dimensional analog of the Jacobian and can be used to compute $D\mathcal{T}[u] = I - DK[u]$ order by order. The operator-norm bound \eqref{eq:DA_inv} is the quantitative companion to this representation.
\end{remark}

\section{A Nonlinear PDAE Example}\label{sec:example}

We illustrate the theory on a simple partial differential-algebraic equation (PDAE) in which the nonlinearity $F[u]$ is obtained as a closed-form functional of $u$ through an auxiliary ordinary differential equation in the spatial variable. Consider the plant
\begin{subequations}\label{eq:pdae_example}
\begin{align}
u_t(x,t) &= u_x(x,t) + \tfrac12\, v(x,t)^2, && x\in[0,1),\ t\ge 0, \label{eq:pdae_u_example}\\
v_x(x,t) &= u(x,t), && v(0,t)=0, \label{eq:pdae_v_example}\\
u(1,t) &= U(t), && u(x,0)=u_0(x)\in L^2(0,1). \label{eq:pdae_bc_example}
\end{align}
\end{subequations}
The algebraic constraint \eqref{eq:pdae_v_example} admits the explicit solution $v(x,t) = \int_0^x u(\xi,t)\,d\xi$, so that the nonlinearity in \eqref{eq:pdae_u_example} is given in closed form by
\begin{equation}\label{eq:F_closed_example}
F[u](x,t) \;=\; \tfrac12\, v(x,t)^2 \;=\; \tfrac12\!\left(\int_0^x u(\xi,t)\,d\xi\right)^{\!2}.
\end{equation}
The identity
\begin{equation}\label{eq:square_identity_example}
\tfrac12\!\left(\int_0^x u(\xi,t)\,d\xi\right)^{\!2}
\;=\; \int_0^x\!\!\int_0^{\xi_1}u(\xi_1,t)\,u(\xi_2,t)\,d\xi_2\,d\xi_1
\end{equation}
recasts \eqref{eq:F_closed_example} into the Volterra form \eqref{eq:VolterraF}, with kernels
\begin{equation}\label{eq:fn_example}
f_2 \equiv 1 \ \text{on}\ T_2(x),\qquad f_n \equiv 0 \ \text{for}\ n\ge 3.
\end{equation}
The nonlinearity is thus a polynomial of degree two in $u$, and Assumption \ref{ass:weighted_growth} is satisfied with $D_f = 1$ and $\rho_f$ arbitrary.

\begin{proposition}
[Backstepping kernels at orders two and three]\label{fact:k2_k3_example}
For the plant \eqref{eq:pdae_example}, the first two nontrivial kernels of the Volterra backstepping transformation \eqref{eq:K_series} are
\begin{align}
k_2(x,\xi_1,\xi_2) &= -\xi_2, \label{eq:k2_example}\\
k_3(x,\xi_1,\xi_2,\xi_3) &= -\xi_3\!\left[(x-\xi_1)(\xi_1+\xi_2) + \tfrac12(\xi_1^2-\xi_2^2) - \tfrac12\xi_3(x-\xi_2)\right]. \label{eq:k3_example}
\end{align}
\end{proposition}

\begin{proof}
The kernel $k_2$ is obtained from the characteristic representation \eqref{eq:kernel_integral} with $n=2$ and $f_2\equiv 1$:
\begin{equation}
k_2(x,\xi_1,\xi_2) \;=\; -\int_0^{\xi_2}f_2(x-\xi_2+s,\xi_1-\xi_2+s,s)\,ds \;=\; -\xi_2,
\end{equation}
which is \eqref{eq:k2_example}.

Since $f_3\equiv 0$, the kernel equation \eqref{eq:kernel_pde} at $n=3$ with inflow condition \eqref{eq:kernel_bc} reduces to
\begin{equation}\label{eq:k3_pde_example}
\partial_x k_3 + \sum_{i=1}^{3}\partial_{\xi_i}k_3 \;=\; B^2_3[k_2,f_2](x,\xi_1,\xi_2,\xi_3),
\qquad k_3(x,\xi_1,\xi_2,0)=0.
\end{equation}
Substituting the already-computed kernel $k_2(x,\xi_1,\xi_2)=-\xi_2$ from \eqref{eq:k2_example} and $f_2\equiv 1$ from \eqref{eq:fn_example} into the definition \eqref{eq:Cnm_def} of $B^2_3[k_2,f_2]$ and carrying out the two inner integrations and the $D^{2,2}_j$ permutation sums yields
\begin{equation}\label{eq:B23_example}
B^2_3[k_2,f_2](x,\xi_1,\xi_2,\xi_3) \;=\; -(x-\xi_1)(\xi_1+\xi_2+\xi_3) - \tfrac12(\xi_1^2-\xi_2^2).
\end{equation}
The inflow condition in \eqref{eq:k3_pde_example} holds for \eqref{eq:k3_example} by the $\xi_3$-prefactor. A direct computation of the four partial derivatives of \eqref{eq:k3_example} and summation yields
\begin{equation}
\partial_x k_3 + \sum_{i=1}^{3}\partial_{\xi_i}k_3 \;=\; -(x-\xi_1)(\xi_1+\xi_2+\xi_3) - \tfrac12(\xi_1^2-\xi_2^2),
\end{equation}
which coincides with \eqref{eq:B23_example}.
\end{proof}

The kernels \eqref{eq:k2_example}--\eqref{eq:k3_example} determine the first two nontrivial terms of the feedback law \eqref{eq:control}:
\begin{align}
U(t) \;=\;
&-\int_0^1\!\!\int_0^{\xi_1}\xi_2\,u(\xi_1,t)u(\xi_2,t)\,d\xi_2\,d\xi_1 \nonumber\\
&-\int_{T_3(1)}\xi_3\!\left[(1-\xi_1)(\xi_1+\xi_2) + \tfrac12(\xi_1^2-\xi_2^2) - \tfrac12\xi_3(1-\xi_2)\right]u(\xi_1,t)u(\xi_2,t)u(\xi_3,t)\,d\xi_3 d\xi_2 d\xi_1 \nonumber\\
&+ \sum_{n\ge 4}\int_{T_n(1)} k_n(1,\xi_1,\ldots,\xi_n)\prod_{i=1}^n u(\xi_i,t)\,d\xi_n\cdots d\xi_1.\label{eq:U_example}
\end{align}
Although $F[u]$ is a single quadratic Volterra term, the feedback-linearizing transformation is an infinite Volterra series: the higher-order kernels $k_n$ for $n\ge 4$ are nonzero and are generated recursively through \eqref{eq:kernel_integral} by the coupling $B^2_n[k_{n-1},f_2]$, with $f_2\equiv 1$ the sole nontrivial plant kernel. Convergence of the series \eqref{eq:U_example} on $L^2(0,1)$ is guaranteed by Theorem \ref{thm:direct_sup} and Corollary \ref{cor:explicit_radius}, with $\rho_K$ arbitrarily large by virtue of \eqref{eq:fn_example}.

\begin{figure}[t]
  \centering
  \includegraphics[width=0.37\linewidth]{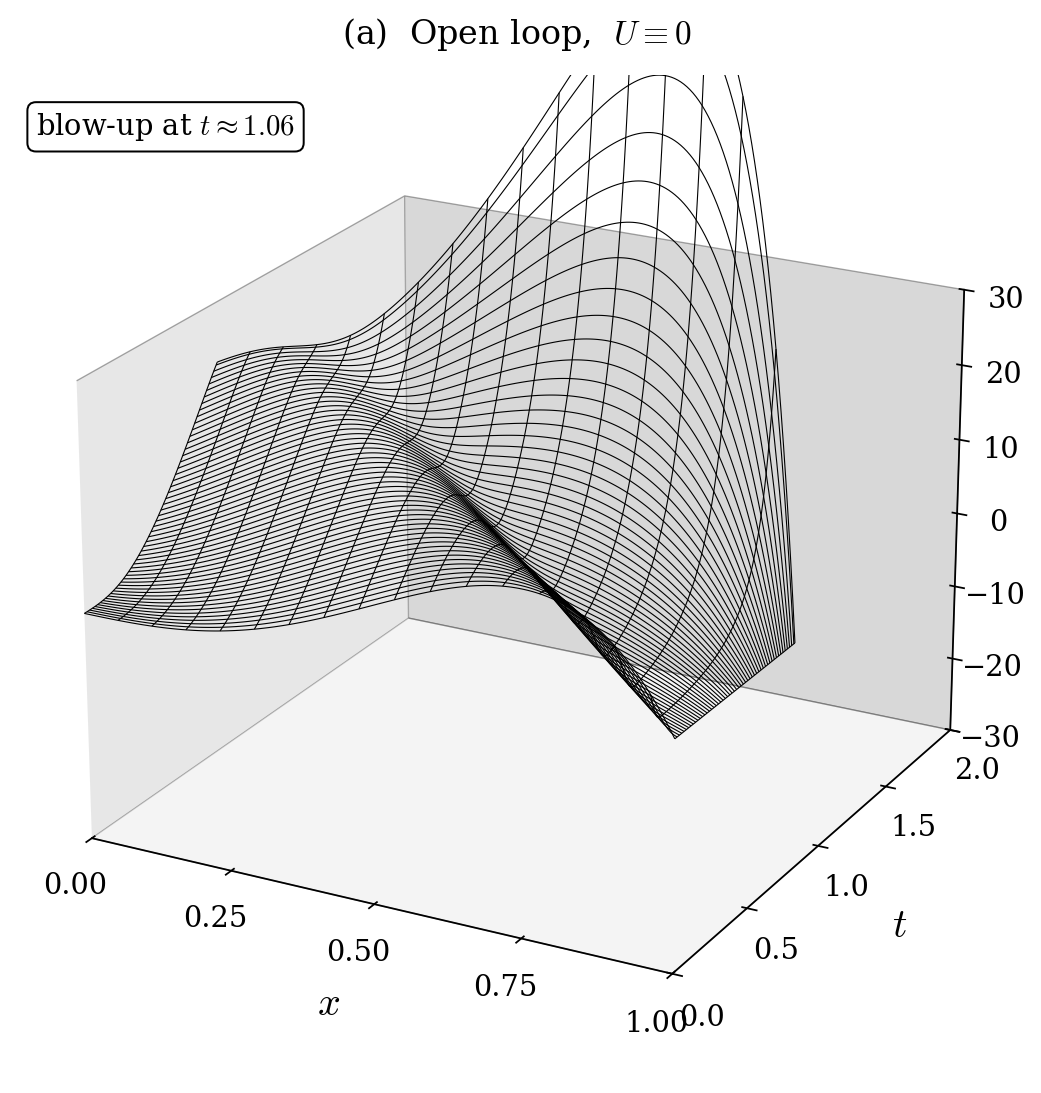}\\[2pt]
  \includegraphics[width=0.37\linewidth]{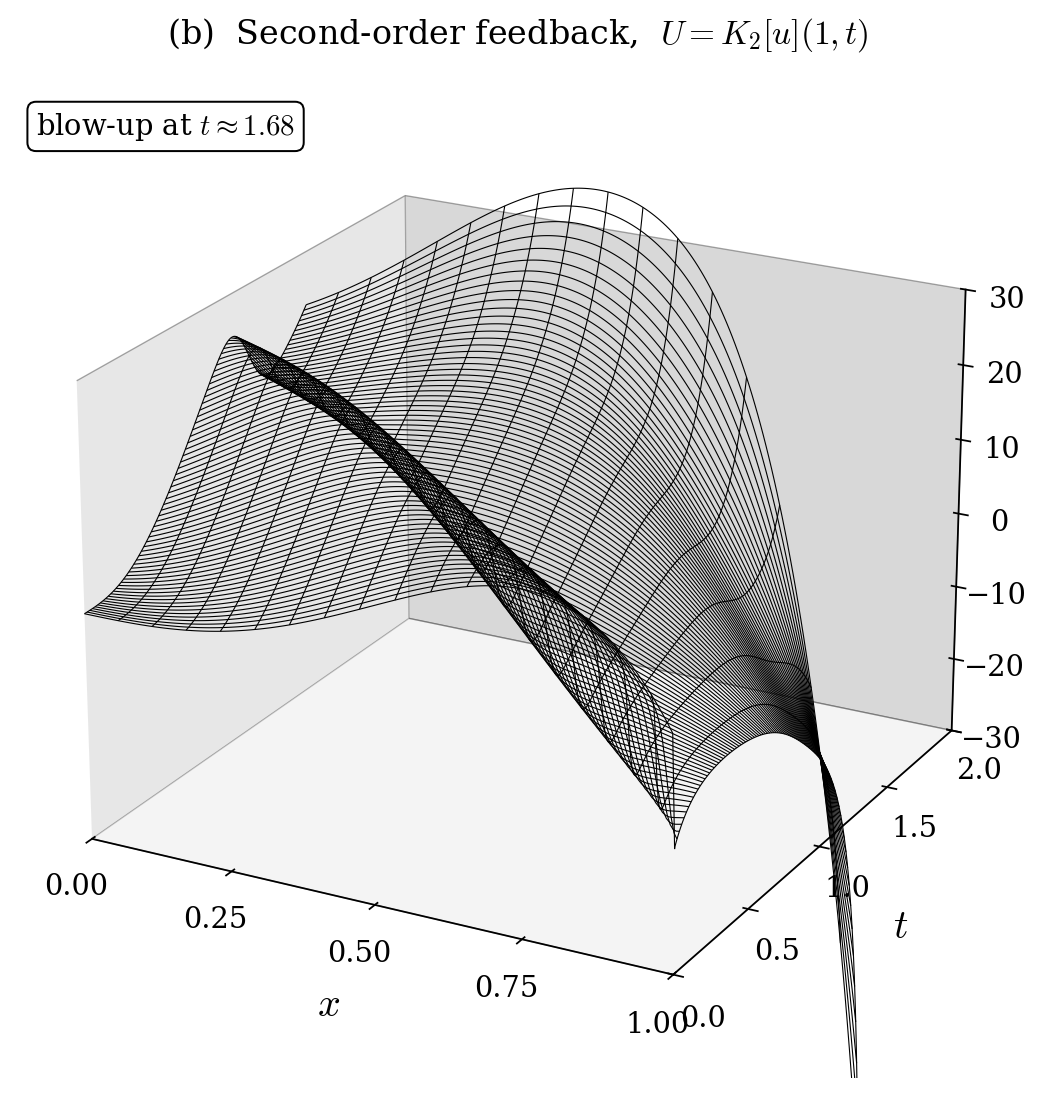}
  \includegraphics[width=0.37\linewidth]{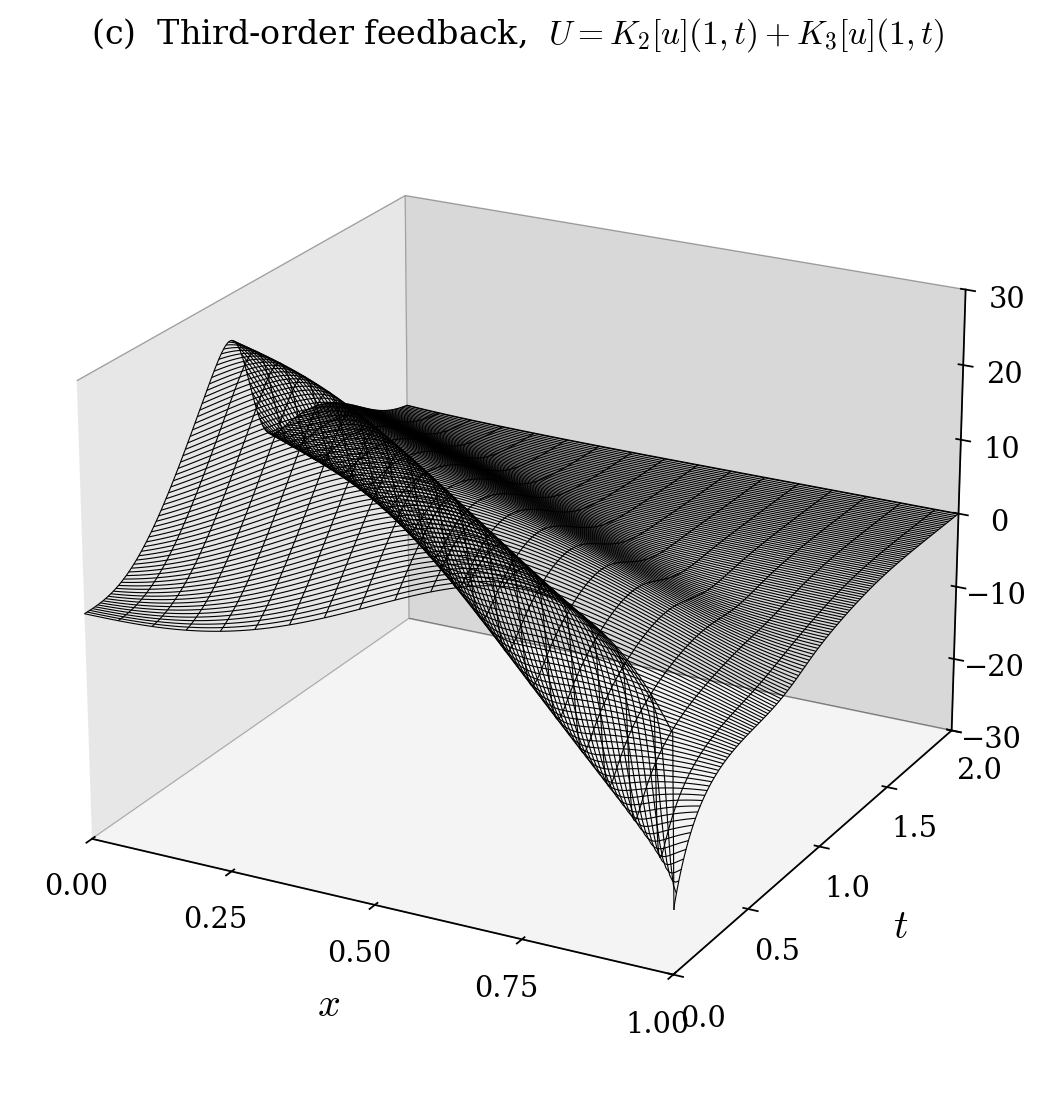}
  \caption{Surface plots of the solution $u(x,t)$ of the plant \eqref{eq:pdae_example}
    with initial condition \eqref{eq:u0_simulation} under three control laws.
    (a) Open loop, $U\equiv 0$: the nonlinearity \eqref{eq:F_closed_example} drives
    the solution to blow up at $t\approx 1.06$.
    (b) Second-order feedback $U=K_2[u](1,t)$ with $k_2$ given by \eqref{eq:k2_example}:
    blow-up is deferred but not prevented, occurring at $t\approx 1.68$.
    (c) Third-order feedback $U=K_2[u](1,t)+K_3[u](1,t)$ with $k_3$ given by
    \eqref{eq:k3_example}: the closed-loop system is stabilized and decays
    toward the origin.}
  \label{fig:simulation}
\end{figure}

We illustrate Proposition \ref{fact:k2_k3_example} numerically by simulating the plant \eqref{eq:pdae_example} with the initial condition
\begin{equation}\label{eq:u0_simulation}
u_0(x) = 140\,x^3(1-x),
\end{equation}
which is positive on $(0,1)$, vanishes at both endpoints, and is concentrated toward $x=1$ (attaining its maximum at $x=3/4$). The spatial domain is discretized on a uniform grid of $201$ points, the transport term $u_x$ is evaluated by upwind finite differences, and the time stepping is carried out by Heun's method with CFL $0.5$. At each step, the constraint \eqref{eq:pdae_v_example} is integrated by the trapezoidal rule to obtain $v(x,t)$, and the feedback law \eqref{eq:control} is evaluated at $x=1$ by quadrature of the kernels \eqref{eq:k2_example} and \eqref{eq:k3_example} on the simplices $T_2(1)$ and $T_3(1)$. Three controllers are compared on the same initial condition: the open-loop case $U\equiv 0$, the second-order truncation $U(t) = K_2[u](1,t)$, and the third-order truncation $U(t) = K_2[u](1,t) + K_3[u](1,t)$.

The results are shown in Fig.~\ref{fig:simulation}. In panel~(a), the open-loop evolution develops a singularity at $t\approx 1.06$: the quadratic self-amplification $\tfrac12 v^2$ in \eqref{eq:pdae_u_example} overwhelms the stabilizing outflow through $x=0$, and the solution escapes every bounded set in finite time. In panel~(b), the second-order feedback injects a strongly negative boundary value $u(1,t)$ that propagates leftward and partially cancels the incoming source; this defers the blow-up to $t\approx 1.68$, but the truncation at order two does not account for the cubic Volterra coupling $B^2_3[k_2,f_2]$ in \eqref{eq:B23_example}, and the residual of order three eventually reasserts the instability. In panel~(c), augmenting the feedback with the order-three term $K_3[u](1,t)$---generated by the kernel \eqref{eq:k3_example}---removes this residual at the leading cubic order, and the closed-loop solution decays to a small neighborhood of the origin: at the end of the simulation window, $\|u(\cdot,2)\|_{L^2(0,1)} \approx 0.20$, down from the peak $\max_{x,t} |u(x,t)| \approx 24$. The contrast between panels~(b) and~(c) is exactly the quantitative content of the Volterra series feedback-linearization program: each successive order of the kernel expansion removes a correspondingly higher-order nonlinear obstruction, and truncations of the feedback law converge to the full linearizing control as the order tends to infinity.

\section{PDE-Free (Quadrature-Based) Feedback Linearization for Transport-Adapted Chen-Fliess Subclass}\label{sec:pde-free}

The controller $U(t)=K[u](1,t)$ of Theorem~\ref{thm:closed_loop} is specified
through the kernels $k_n$, which are constructed in Section~\ref{sec:kernels}
from the integral equation representation \eqref{eq:kernel_integral} of the
transport PDE \eqref{eq:kernel_pde}--\eqref{eq:kernel_bc} on the simplex
$T_n(1)$. We show in this section that, for a \emph{narrower} class of plants
--- those whose Volterra kernels $f_n$ admit a polynomial gap-basis expansion
with analytic, geometrically decaying coefficients --- the very same controller
is obtained without any kernel PDEs arising: the sequence of kernels reduces to a triangular system of scalar first-order ODEs on $[0,1]$, each solvable by one quadrature. Throughout this section the convention
$\xi_0 := x$ of Section~\ref{sec:kernels} is used.

The representation of $F[u]$ used in this section is the Chen--Fliess \cite{fliess1981fonctionnelles} form of
the spatially-causal Volterra nonlinearity \eqref{eq:VolterraF}, in which each Volterra kernel is encoded by a family of space-dependent scalar
coefficients multiplying a polynomial basis built from the consecutive
differences $\xi_0-\xi_1$, $\xi_1-\xi_2$, $\ldots$, $\xi_{n-1}-\xi_n$ (with
$\xi_0:=x$) on the simplex $T_n(x)$. 

In order to avoid having to solve PDEs to generate the kernels of the backstepping operator $K$, in this paper we narrow our attention to the following particular ``transport-adapted'' subclass of Chen-Fliess series,  
\begin{equation}\label{eq:CF-series}
F[u](x,t)
\;=\;
\sum_{n=2}^{\infty}\,\sum_{p\in\N_0^n}
\int_{T_n(x)} b^{(n)}_p(x)
\prod_{j=0}^{n-1}\frac{(\xi_j-\xi_{j+1})^{p_j}}{p_j!}
\prod_{i=1}^{n} u(\xi_i,t)\,d\xi_n\cdots d\xi_1,
\qquad \xi_0:=x,
\end{equation}
in which $n$ is the order of the nonlinearity,
$p=(p_0,\ldots,p_{n-1})$ is a tuple of nonnegative integer exponents, the
space-dependent functions $b^{(n)}_p(x)$ are the scalar coefficients of the
representation, and each integrand is a product of consecutive differences
$\xi_j-\xi_{j+1}$ raised to the powers $p_j$ and of the values of $u$ at the
integration points $\xi_1,\ldots,\xi_n$. 

The feedback linearization design in this section for the class \eqref{eq:CF-series} can be generalized from  $b^{(n)}_p(x)$ to $b^{(n)}_p(x,x-\xi_n)$. However, this design generalization, detailed later in Remark \ref{rem:recursion-invariant-cf}, comes with a prohibitive extra cost in convergence analysis, so we stick with $b^{(n)}_p(x)$ in \eqref{eq:CF-series}.

The Volterra example of
Section~\ref{sec:example}, with a single Volterra kernel $f_2\equiv 1$,
happens to also admit the transport-adapted Chen-Fliess subclass representation \eqref{eq:CF-series}, with scalar coefficient
$b^{(2)}_{(0,0)}(x)\equiv 1$ and all other $b^{(n)}_p$ zero.

\begin{remark}[Fréchet vs.\ causal derivatives]\label{rem:Frechet_vs_causal}
The Fréchet differentiability established in Proposition~\ref{prop:diffeo} is the appropriate notion for treating the backstepping map $\mathcal{T}=I-K$ as a local diffeomorphism on $L^2(0,1)$ and underlies the invertibility and stability results of Sections~\ref{sec:inverse}, \ref{sec:closed-loop}, and \ref{sec-exist}; it acts on the functional dependence $u\mapsto \mathcal{T}[u]$ without regard to ordering in the spatial variable. When the nonlinear operator admits a Chen--Fliess representation, however, the Volterra simplices carry the intrinsic ordering $0\le \xi_n\le\cdots\le\xi_1\le x$, aligned with the transport direction, and the natural infinitesimal variation is the extension of the input by a terminal segment in $x$ --- Fliess's causal derivation~\cite{fliess1986derivation} --- under which the noncommutative generating series~\cite{fliess1981fonctionnelles} becomes a genuine Taylor expansion. The reduction in this section to a triangular cascade of scalar ODEs is compatible with this structure: the transport operator leaves the gap variables invariant, so the causal derivative respects the cascade. Thus, the Fréchet framework of Section~\ref{sec:diffeo} provides the correct global analytic setting, while the Chen--Fliess specialization reveals an additional algebraic--causal structure aligned with the direction of propagation.
\end{remark}

\subsection{Gap basis}\label{sec:gap-basis}

For $n\ge 1$ and $p=(p_0,\ldots,p_{n-1})\in\N_0^n$, define the polynomial
\emph{gap basis} on $T_n(x)$ by
\begin{equation}\label{eq:Phi-def}
\Phi_p(x;\xi_1,\ldots,\xi_n)
:=
\prod_{j=0}^{n-1}\frac{(\xi_j-\xi_{j+1})^{p_j}}{p_j!},
\qquad \xi_0:=x.
\end{equation}
Let
\begin{equation}\label{eq:Tn-op}
\mathsf{T}_n \;:=\; \partial_x + \sum_{i=1}^n \partial_{\xi_i}
\end{equation}
denote the transport operator on the left-hand side of \eqref{eq:kernel_pde}.
The generator $\mathsf T_n$ is the derivative along the flow that advances
each of the variables $x,\xi_1,\ldots,\xi_n$ by the same parameter $s$, and
each gap $\xi_j-\xi_{j+1}$ is invariant under this flow. Hence
\begin{equation}\label{eq:TPhi-zero}
\mathsf T_n\,\Phi_p(x;\xi) \;=\; 0
\qquad\text{for every }p\in\N_0^n.
\end{equation}
This transport-invariance is what is  behind the reduction of
the kernel PDE to scalar ODEs.

\subsection{Transport-adapted subclass}\label{sec:cf-class}

Expand each plant kernel $f_n$ in the gap basis,
\begin{equation}\label{eq:fn-gap}
f_n(x,\xi_1,\ldots,\xi_n)
\;=\;
\sum_{p\in\N_0^n} b^{(n)}_p(x)\,\Phi_p(x;\xi),
\qquad n\ge 2,
\end{equation}
and introduce

\begin{assumption}[Transport-adapted Chen-Fliess subclass]\label{ass:cf}
The kernels $f_n$ in \eqref{eq:VolterraF} are jointly real-analytic on
$T_n(1)$, and their gap-basis coefficients $b^{(n)}_p$ in \eqref{eq:fn-gap}
are real-analytic on $[0,1]$. There exist constants $D>0$, $\rho>0$, and
$\mu,\nu>0$ such that
\begin{equation}\label{eq:cf-bound}
\sup_{x\in[0,1]}\bigl|\bigl(b^{(n)}_p\bigr)^{(\sigma)}(x)\bigr|
\;\le\;
D\,\rho^{-(n-1)}\,\mu^{-|p|}\,\nu^{-\sigma},
\qquad n\ge 2,\;p\in\N_0^n,\;\sigma\in\N_0,
\end{equation}
where $|p|=p_0+\cdots+p_{n-1}$.
\end{assumption}

Assumption~\ref{ass:cf} strengthens Assumption~\ref{ass:weighted_growth} in
two ways: analyticity of each $b^{(n)}_p$ in $x$ and geometric decay in
$|p|$. The four parameters play distinct roles: $D$ is an absolute scale;
$\rho$ governs decay in the Volterra order $n$; $\mu$ governs decay in the
gap-exponent $|p|$; and $\nu$ is a Cauchy-type radius governing derivative
growth in $\sigma$ (the bound on $(b^{(n)}_p)^{(\sigma)}$ grows no faster
than $\nu^{-\sigma}$, consistent with the natural estimate for real-analytic
functions on $[0,1]$). Plants with kernels $f_n$ polynomial in
$\xi_1,\ldots,\xi_n$ of bounded total gap-degree satisfy \eqref{eq:cf-bound}
with arbitrarily large $\mu$; this includes the PDAE of
Section~\ref{sec:example}, for which $f_2\equiv 1$ and $f_n\equiv 0$ for
$n\ge 3$.

\begin{remark}\label{rem:cf-implies-weighted}
Under Assumption~\ref{ass:cf}, on $T_n(1)$ the $n$ gaps
$\Delta_r:=\xi_r-\xi_{r+1}$, $r=0,\ldots,n-1$ (with $\xi_0:=x$,
$\xi_{n+1}:=0$ absent), satisfy $\sum_{r=0}^{n-1}\Delta_r\le x\le 1$.
Expanding $\Phi_p(x;\xi)$ in the gap variables and summing over
$p\in\N_0^n$ gives
\begin{equation}\label{eq:cf-sup-bound}
\|f_n\|_{L^\infty(T_n(1))}
\;\le\;
\sum_{p\in\N_0^n}D\,\rho^{-(n-1)}\mu^{-|p|}\,\prod_{r=0}^{n-1}\frac{\Delta_r^{p_r}}{p_r!}
\;=\;
D\,\rho^{-(n-1)}\,\exp\!\Bigl(\frac{\sum_{r=0}^{n-1}\Delta_r}{\mu}\Bigr)
\;\le\;
D\,\rho^{-(n-1)}\,e^{\,1/\mu},
\end{equation}
uniformly in $n$. Hence Assumption~\ref{ass:weighted_growth} is
recovered with $D_f=D\,e^{\,1/\mu}$ and $\rho_f=\rho$.
\end{remark}

\subsection{Cascade of scalar ODEs}\label{sec:scalar-ode}

The construction of the kernels $k_n$ of the backstepping operator $K$, by scalar quadrature, rests on the ansatz, for sequences
$\{a^{(n)}_p\}_{p\in\N_0^n}\subset C^1[0,1]$,
\begin{equation}\label{eq:kn-ansatz}
k_n(x,\xi_1,\ldots,\xi_n)
\;:=\;
\sum_{p\in\N_0^n}\bigl[a^{(n)}_p(x)-a^{(n)}_p(x-\xi_n)\bigr]\,\Phi_p(x;\xi).
\end{equation}

\begin{lemma}[Transport identity for the $k_n$ inflow ansatz]\label{lem:inflow-ansatz}
For any $C^1$ sequence $\{a^{(n)}_p\}_{p\in\N_0^n}$, the function $k_n$ in
\eqref{eq:kn-ansatz} satisfies the inflow condition \eqref{eq:kernel_bc}
identically, and
\begin{equation}\label{eq:Tn-kn}
\mathsf T_n\,k_n(x,\xi_1,\ldots,\xi_n)
\;=\;
\sum_{p\in\N_0^n}\bigl(a^{(n)}_p\bigr)'(x)\,\Phi_p(x;\xi).
\end{equation}
\end{lemma}

\begin{proof}
At $\xi_n=0$, $a^{(n)}_p(x)-a^{(n)}_p(x-\xi_n)=a^{(n)}_p(x)-a^{(n)}_p(x)=0$, so
\eqref{eq:kernel_bc} holds. For \eqref{eq:Tn-kn}, $\mathsf T_n\Phi_p=0$ by
\eqref{eq:TPhi-zero}, so $\mathsf T_n k_n$ equals the sum over $p$ of $\Phi_p$
times $\mathsf T_n[a^{(n)}_p(x)-a^{(n)}_p(x-\xi_n)]$. Componentwise,
\begin{subequations}
\begin{align}
\partial_x\bigl[a^{(n)}_p(x)-a^{(n)}_p(x-\xi_n)\bigr]
&= \bigl(a^{(n)}_p\bigr)'(x)-\bigl(a^{(n)}_p\bigr)'(x-\xi_n),\\
\partial_{\xi_i}\bigl[a^{(n)}_p(x)-a^{(n)}_p(x-\xi_n)\bigr]
&= 0,\qquad 1\le i\le n-1,\\
\partial_{\xi_n}\bigl[a^{(n)}_p(x)-a^{(n)}_p(x-\xi_n)\bigr]
&= \bigl(a^{(n)}_p\bigr)'(x-\xi_n).
\end{align}
\end{subequations}
Summing, the two copies of $(a^{(n)}_p)'(x-\xi_n)$ cancel, leaving
$(a^{(n)}_p)'(x)$.
\end{proof}

Under the ansatz \eqref{eq:kn-ansatz}, the kernel PDE
\eqref{eq:kernel_pde}--\eqref{eq:kernel_bc} is equivalent, by
Lemma~\ref{lem:inflow-ansatz}, to the gap-basis identity
\begin{equation}\label{eq:matching-gap}
\sum_p\bigl(a^{(n)}_p\bigr)'(x)\,\Phi_p(x;\xi)
\;=\;
-\sum_p b^{(n)}_p(x)\,\Phi_p(x;\xi)
+\sum_{m=2}^{n-1} B^m_n[k_{n-m+1},f_m](x,\xi).
\end{equation}
(The $m=n$ term in the sum on the right-hand side of \eqref{eq:kernel_pde}
vanishes, since $k_{n-m+1}\vert_{m=n}=k_1\equiv 0$ renders
$B^n_n[k_1,f_n]\equiv 0$ by \eqref{eq:Cnm_def}. The effective range in
\eqref{eq:matching-gap} is therefore $m\in\{2,\ldots,n-1\}$.) Projecting \eqref{eq:matching-gap} onto
$\Phi_P$ yields, for each $P\in\N_0^n$, the scalar ODE
\begin{equation}\label{eq:scalar-ode}
\bigl(a^{(n)}_P\bigr)'(x)
\;=\;
-\,b^{(n)}_P(x) + c^{(n)}_P(x),
\qquad
a^{(n)}_P(0)=0,
\end{equation}
where $c^{(n)}_P(x)$ is the $\Phi_P$-coefficient of the coupling sum on the
right-hand side of \eqref{eq:matching-gap}. The normalization
$a^{(n)}_P(0)=0$ fixes the residual additive constant in \eqref{eq:kn-ansatz}.
The solution is the quadrature
\begin{equation}\label{eq:scalar-quad}
a^{(n)}_P(x)
\;=\;
\int_0^x \bigl[-b^{(n)}_P(s)+c^{(n)}_P(s)\bigr]\,ds.
\end{equation}

\begin{remark}[Characteristic integral representation vs.\ scalar quadrature]
\label{rem:quadrature-vs-characteristic}
Both the general Volterra construction and the transport-adapted Chen--Fliess specialization share a key feature: neither requires solving kernel PDEs. In the general case, the kernels are obtained from the characteristic integral representation \eqref{eq:kernel_integral}, which expresses $k_n$ explicitly through integration along transport characteristics, i.e., as a parametrized (or parametric) quadrature in the variable $s$, whereas in the Chen--Fliess specialization this structure is further reduced to scalar quadratures such as \eqref{eq:scalar-quad}. The difference lies not in the presence of integrals, but in what is being integrated: in \eqref{eq:kernel_integral}, each evaluation of $k_n(x,\xi_1,\ldots,\xi_n)$ involves lower-order kernels evaluated at shifted arguments, so the computation remains at the level of functions defined on the simplex $T_n$ and retains its high-dimensional functional dependence; in contrast, after projection onto the gap basis, the Chen--Fliess case reduces the recursion to scalar coefficient functions $a_P^{(n)}(x)$ governed by one-dimensional ODEs, whose solutions are obtained by quadrature. Thus, while both approaches avoid PDE solvers, the former preserves a high-dimensional functional structure, whereas the latter achieves a complete scalarization of the construction.
\end{remark}

The content of the construction \eqref{eq:scalar-quad} lies in exhibiting $c^{(n)}_P$ as an
explicit finite combinatorial expression in the data of lower order for the Chen-Fliess kernels.

\begin{lemma}[Explicit formula for $c^{(n)}_P$]\label{lem:cnP}
Let $n\ge 3$ and $P\in\N_0^n$. Under Assumption~\ref{ass:cf}, the
$\Phi_P$-coefficient $c^{(n)}_P(x)$ of the coupling $\sum_{m=2}^{n-1} B^m_n[k_{n-m+1},f_m](x,\xi)$ in  \eqref{eq:matching-gap} admits the formal expansion
\begin{equation}\label{eq:cnP-form}
c^{(n)}_P(x)
\;=\;
\sum_{m=2}^{n-1}\;
\sum_{\substack{q\in\N_0^{\,n-m+1},\;q'\in\N_0^{\,m}\\\sigma,\tau\in\N_0,\;\alpha\in\N_0^{\,n}\\|\alpha|\le\tau,\;|q|+|q'|+\sigma+|\alpha|=|P|-1}}
\Gamma^{(n,m)}_{P;q,q',\sigma,\tau,\alpha}\;
\frac{x^{\tau-|\alpha|}}{(\tau-|\alpha|)!}\,
\bigl(a^{(n-m+1)}_q\bigr)^{(\tau)}(x)\;
\bigl(b^{(m)}_{q'}\bigr)^{(\sigma)}(x),
\end{equation}
where the integer-valued combinatorial coefficients
$\Gamma^{(n,m)}_{P;q,q',\sigma,\tau,\alpha}\in\mathbb{Z}$ are
constructed in Appendix~\ref{app:Gamma} and depend only on
$(n,m,P,q,q',\sigma,\tau,\alpha)$ and not on the plant data. For fixed $P$ the sum is infinite in $\tau$. 
\end{lemma}

\begin{proof}
Appendix~\ref{app:Gamma} substitutes the ansatz \eqref{eq:kn-ansatz}
for $k_{n-m+1}$ and the gap-basis expansion \eqref{eq:fn-gap} for
$f_m$ into the definition \eqref{eq:Cnm_def} of
$B^m_n[k_{n-m+1},f_m]$, Taylor-expands the coefficients
$a^{(n-m+1)}_q(x-s)$ and $b^{(m)}_{q'}(s)$ around $x$, expands each
power $s^\tau$ produced by the $a$-Taylor series as a polynomial in
the gap variables of the extended chain via $s=x-(x-s)$, and extracts
the $\Phi_P$-coefficient. The pure-$x$ factor $x^{\tau-|\alpha|}$
carries a residual $1/(\tau-|\alpha|)!$ from the Taylor/multinomial
bookkeeping — not cleared by the gap-basis normalization, since $x$
is not a gap variable — which we absorb into the divided-power form
$x^{\tau-|\alpha|}/(\tau-|\alpha|)!$ displayed in \eqref{eq:cnP-form}.
With this absorption, all factorials cancel and the residual
coefficient $\Gamma^{(n,m)}_{P;q,q',\sigma,\tau,\alpha}$ extracted in
the appendix is a signed integer. The multi-index $\alpha\in\N_0^n$
records the distribution of the $s^\tau$ expansion across the
available gaps. At this stage \eqref{eq:cnP-form} is a formal coefficient identity: the signed integers
$\Gamma^{(n,m)}_{P;q,q',\sigma,\tau,\alpha}$ are defined, and the identity holds as an
equality of finite-$\tau$ partial sums for each $\tau_{\max}\in\N_0$. Absolute convergence
of the full series in $(q,q',\sigma,\tau,\alpha)$ on $[0,1]$ — which promotes
\eqref{eq:cnP-form} from a formal identity to an equality of functions — is established
separately in the proof of Lemma~\ref{lem:kn-convergence} via the divided-power algebra
construction of Section~\ref{sec:gap-basis-convergence}.
\end{proof}

Equation \eqref{eq:cnP-form} shows the cascade is triangular: $a^{(n)}_P$ depends only on lower-order cascade data
$\{a^{(m)}_q:\,m<n\}$ and plant data $\{b^{(k)}_{q'}:\,k\le n-1\}$,
through derivatives of all orders $\tau,\sigma\in\N_0$. The cascade closes
at $n=2$: the coupling sum in \eqref{eq:matching-gap} is empty ($m$ would
be restricted to $\{2,\ldots,1\}=\varnothing$), so $c^{(2)}_P\equiv 0$ and
\begin{equation}\label{eq:a2-quad}
a^{(2)}_P(x) \;=\; -\int_0^x b^{(2)}_P(s)\,ds,\qquad P\in\N_0^2.
\end{equation}

\subsection{Convergence of the gap-basis series}\label{sec:gap-basis-convergence}

The ansatz \eqref{eq:kn-ansatz} writes each $k_n$ as an infinite sum over
$p\in\N_0^n$. We now establish that this sum, together with the coupling
sums generated by the operators $B^m_n$ along the cascade, converges
absolutely and uniformly on $T_n(1)$ at every order $n$. The argument
proceeds by exhibiting the plant data and the cascade coefficients in an
analytic divided-power coefficient algebra on which the cascade operators
act continuously. The convergence of the infinite Volterra feedback
\eqref{eq:U-pde-free} itself is inherited from Theorem~\ref{thm:direct_sup}
through Remark~\ref{rem:cf-implies-weighted}.


\begin{definition}[Divided-power coefficient family]\label{def:dp-algebra}
For $N\ge 1$, $r>0$, and $R>0$, and for a family
$a=\{a_p\}_{p\in\N_0^N}$ of functions $a_p\in C^\infty([0,1])$, define the
norm
\begin{equation}\label{eq:dp-norm}
\|a\|_{N,r,R}
\;:=\;
\sup_{x\in[0,1]}
\sum_{p\in\N_0^N}\sum_{\sigma=0}^{\infty}
\bigl|(a_p)^{(\sigma)}(x)\bigr|\,
\frac{r^{|p|}\,R^{\sigma}}{p!\,\sigma!},
\end{equation}
where $p!:=p_0!\cdots p_{N-1}!$ and $|p|:=p_0+\cdots+p_{N-1}$. Let
$\mathcal A_{N,r,R}$ denote the set of families $a$ with
$\|a\|_{N,r,R}<\infty$.
\end{definition}

The norm $\|\cdot\|_{N,r,R}$ has four properties governing the cascade —
absolute uniform convergence on $T_N(1)$, a Banach-algebra inequality
under divided-power multiplication, continuity under linear substitution
of gap variables, and continuity under split-gap integration — together
with the plant-data estimate
\begin{equation}\label{eq:b-norm-finite}
\|b^{(n)}\|_{n,r,R}
\;\le\;
D\,\rho^{-(n-1)}\,\exp\!\Bigl(\frac{n\,r}{\mu}\Bigr)\,
\exp\!\Bigl(\frac{R}{\nu}\Bigr)
\;<\;\infty,
\qquad n\ge 2,
\end{equation}
valid for every $r>0$ and $R>0$. All five are recorded, with proofs,
as Lemma~\ref{lem:dp-algebra} in Appendix~\ref{app:dp-algebra}; they
drive the convergence argument that follows.


\begin{lemma}[Convergence of the gap-basis series]\label{lem:kn-convergence}
Let Assumption~\ref{ass:cf} hold. For every $n\ge 2$, the coefficient
family $a^{(n)}=\{a^{(n)}_P\}_{P\in\N_0^n}$ defined by the cascade
\eqref{eq:a-recursion}, \eqref{eq:cnP-form} satisfies
$a^{(n)}\in\mathcal A_{n,r,R}$ for every $r\ge 1$ and every $R>0$.
Consequently, the series \eqref{eq:kn-ansatz} defining $k_n$ converges
absolutely and uniformly on $T_n(1)$, the limit is real-analytic on
$T_n(1)$, and the inflow condition \eqref{eq:kernel_bc} holds.
\end{lemma}

\begin{proof}
Induction on $n\ge 2$. Fix $r\ge 1$ and $R>0$.

\emph{Base case, $n=2$.} From \eqref{eq:a2-quad},
$a^{(2)}_P(x)=-\int_0^x b^{(2)}_P(s)\,ds$. For $\sigma\ge 1$,
$(a^{(2)}_P)^{(\sigma)}(x)=-(b^{(2)}_P)^{(\sigma-1)}(x)$; for $\sigma=0$,
$|a^{(2)}_P(x)|\le\sup_{s\in[0,1]}|b^{(2)}_P(s)|$. Substituting
\eqref{eq:cf-bound} and summing,
\begin{equation}\label{eq:a2-norm}
\|a^{(2)}\|_{2,r,R}
\;\le\;
(R+1)\,\|b^{(2)}\|_{2,r,R}
\;\le\;
(R+1)\,D\,\rho^{-1}\,e^{2r/\mu}\,e^{R/\nu}
\;<\;\infty.
\end{equation}

\emph{Inductive step, $n\ge 3$.} Suppose, for every $m\in\{2,\ldots,n-1\}$
and every $r\ge 1$, $R>0$, that $a^{(m)}\in\mathcal A_{m,r,R}$.
Fix the target radii $(r,R)$ at level $n$. Each term of $B^m_n[k_k,f_m]$
with $k=n-m+1$ is obtained from $a^{(k)}$ and $b^{(m)}$ by the composition
of the following operations, each continuous in the divided-power norm:

(i) Substitution of the arguments of $k_k$: the gap variables of $k_k$,
evaluated at the arguments
$(x,\xi_1,\ldots,\xi_{j-1},s,\xi_j,\ldots,\xi_{k-1})$ in \eqref{eq:Cnm_def},
are nonnegative integer linear combinations of the gap variables of the
extended chain \eqref{eq:extended-chain}, with row sums bounded by a
constant $L_0$ depending only on $n$. Lemma~\ref{lem:dp-algebra}(iii)
replaces $r$ by $L_0 r$ in the input norm.

(ii) Taylor expansion of $b^{(m)}_{q'}(s)$ around $x$, using the shift
$s-x$: by \eqref{eq:taylor-b-app}, the result is a power series in
$(s-x)$ whose coefficients, of order $\sigma$, are
$(b^{(m)}_{q'})^{(\sigma)}(x)/\sigma!$. Under Assumption~\ref{ass:cf} this
Taylor series is entire: the bound \eqref{eq:cf-bound} yields terms
$\le|s-x|^\sigma\nu^{-\sigma}/\sigma!$, which sum to $e^{|s-x|/\nu}$ for
every $s,x$. Writing $(s-x)^\sigma=(-1)^\sigma(x-s)^\sigma$ and noting that
$x-s$ is a nonnegative integer linear combination of gap variables of the
extended chain (by \eqref{eq:s-decomposition}), this is a substitution of
the type covered by Lemma~\ref{lem:dp-algebra}(iii), applied to the
variable $x-s$ written in gap form; the sign $(-1)^\sigma$ is absorbed into
the integer coefficient structure and does not affect the norm estimate.

(iii) Multiplication of the two resulting divided-power series.
Lemma~\ref{lem:dp-algebra}(ii) bounds the norm of the product by the
product of norms.

(iv) The permutation-summation operator $D^{k,m}_j$ is a finite sum of
linear substitutions of the type in Lemma~\ref{lem:dp-algebra}(iii); the
sum over $j$ is finite.

(v) Integration in $s$ over $[\xi_j,\xi_{j-1}]$, by \eqref{eq:s-decomposition},
is the split-gap integration of Lemma~\ref{lem:dp-algebra}(iv), with the
split variable $\theta=\xi_{j-1}-s$ and $\omega=s-\xi_j$. This multiplies
the norm by a factor of $r$.

The outer sum over $m\in\{2,\ldots,n-1\}$ in \eqref{eq:matching-gap} is
finite. Therefore, after finitely many applications of operations
(i)--(v), the radius $r$ is replaced by $r'=r\cdot L_*$ for some
constant $L_*\ge 1$ depending only on $n$, while the radius $R$ is
preserved (the $1/\sigma!$ factor in the norm absorbs the $\nu^{-\sigma}$
growth introduced by Taylor expansion for every $R>0$). This gives
$c^{(n)}\in\mathcal A_{n,r',R}$. Since the inductive hypothesis holds
for \emph{every} $r''\ge 1$ and every $R''>0$, we may apply it at
$r''=r'$ to conclude $a^{(m)}\in\mathcal A_{m,r',R}$ for each
$m\in\{2,\ldots,n-1\}$, and Lemma~\ref{lem:dp-algebra}(ii)--(iv) then
combines these into $c^{(n)}\in\mathcal A_{n,r',R}$. By the inclusion
$\mathcal A_{n,r',R}\subseteq\mathcal A_{n,r,R}$ (which follows from
monotonicity of the norm in $r$, since $r\le r'$ implies
$r^{|p|}\le (r')^{|p|}$ term by term), we obtain
$c^{(n)}\in\mathcal A_{n,r,R}$ for the original target radii.

From \eqref{eq:b-norm-finite}, $b^{(n)}\in\mathcal A_{n,r,R}$. Hence the
family $d^{(n)}:=-b^{(n)}+c^{(n)}\in\mathcal A_{n,r,R}$. The quadrature
$a^{(n)}_P(x)=\int_0^x d^{(n)}_P(s)\,ds$ gives
$(a^{(n)}_P)^{(\sigma)}=(d^{(n)}_P)^{(\sigma-1)}$ for $\sigma\ge 1$, and
$|a^{(n)}_P(x)|\le\sup_{s\in[0,1]}|d^{(n)}_P(s)|$, yielding
\begin{equation}\label{eq:an-norm}
\|a^{(n)}\|_{n,r,R}\;\le\;(R+1)\,\|d^{(n)}\|_{n,r,R}\;<\;\infty.
\end{equation}

\emph{Series convergence and analyticity.} With
$a^{(n)}\in\mathcal A_{n,r,R}$ for $r\ge 1$, Lemma~\ref{lem:dp-algebra}(i)
gives absolute and uniform convergence of
$\sum_p a^{(n)}_p(x)\Phi_p(x;\xi)$ on $T_n(1)$.
The shifted term $\sum_p a^{(n)}_p(x-\xi_n)\Phi_p(x;\xi)$ is handled by
the trivial sup-in-$x$ bound
$|a^{(n)}_p(x-\xi_n)|\le\sup_{y\in[0,1]}|a^{(n)}_p(y)|$, valid since
$x-\xi_n\in[0,1]$ on $T_n(1)$. Replacing $|a^{(n)}_p(x)|$ by this sup in
the $\sigma=0$ contribution to \eqref{eq:dp-norm} yields a finite bound
on $\sum_p \sup_{y\in[0,1]}|a^{(n)}_p(y)|\,r^{|p|}/p!$, controlled by
$\|a^{(n)}\|_{n,r,R}$; combining with Lemma~\ref{lem:dp-algebra}(i) gives
absolute and uniform convergence of the shifted sum on $T_n(1)$.
Hence the ansatz \eqref{eq:kn-ansatz} defines $k_n$ as an absolutely and
uniformly convergent series on $T_n(1)$. The limit is real-analytic: the
finiteness of $\|a^{(n)}\|_{n,r,R}$ for every $R>0$ implies
\begin{equation}\label{eq:all-derivatives-bound}
\sup_{x\in[0,1]}\sum_{p\in\N_0^n}\bigl|(a^{(n)}_p)^{(\sigma)}(x)\bigr|
\,\frac{r^{|p|}}{p!}
\;\le\;
\frac{\sigma!}{R^\sigma}\,\|a^{(n)}\|_{n,r,R},
\qquad \sigma\in\N_0,
\end{equation}
so the $\sigma$-th derivative series $\sum_p (a^{(n)}_p)^{(\sigma)}(x)\,\Phi_p(x;\xi)$
converges absolutely and uniformly on $T_n(1)$ for every $\sigma$
(the same sup-in-$x$ bound handles the shifted derivative
$\sum_p (a^{(n)}_p)^{(\sigma)}(x-\xi_n)\,\Phi_p(x;\xi)$). Hence
all derivatives of the partial sums converge uniformly to the
corresponding derivatives of the limit, and the limit is real-analytic
on $T_n(1)$. The inflow condition \eqref{eq:kernel_bc} follows from
Lemma~\ref{lem:inflow-ansatz}.
\end{proof}

\subsection{Main theorem on feedback linearizing stabilization of hyperbolic PDE with  transport-adapted Chen-Fliess subclass}\label{sec:pde-free-theorem}

We now combine the scalar coefficient cascade with the convergence of the gap-basis series to obtain a feedback-linearizing controller for the transport-adapted Chen–Fliess subclass.

\begin{theorem}[PDE-free feedback linearization of the transport-adapted Chen–Fliess subclass]\label{thm:pde-free}
Let Assumption~\ref{ass:cf} hold. Define the scalar sequence
$\{a^{(n)}_P(x)\}_{n\ge 2,\,P\in\N_0^n}$ recursively in $n$ by
\begin{equation}\label{eq:a-recursion}
a^{(n)}_P(x)
\;:=\;
\int_0^x\bigl[-\,b^{(n)}_P(s)+c^{(n)}_P(s)\bigr]\,ds,
\qquad x\in[0,1],
\end{equation}
with $c^{(n)}_P$ given by \eqref{eq:cnP-form} for $n\ge 3$ and
$c^{(2)}_P\equiv 0$. Define the kernels $\widehat k_n$ by the ansatz
\eqref{eq:kn-ansatz} with these $a^{(n)}_p$, and the feedback
\begin{equation}\label{eq:U-pde-free}
\widehat U(t)
\;:=\;
\sum_{n\ge 2}\int_{T_n(1)}
\widehat k_n(1,\xi_1,\ldots,\xi_n)\prod_{i=1}^n u(\xi_i,t)\,d\xi_n\cdots d\xi_1.
\end{equation}
Then:
\begin{enumerate}[label=\textup{(\alph*)}]
\item {\em Well-definedness via scalar algebra only.} For each fixed
$(n,P)$ with $n\ge 2$ and $P\in\N_0^n$, the coefficient $a^{(n)}_P$ is
obtained by an absolutely convergent scalar coefficient expression
assembled from finitely many algebra operations on the plant data and
lower-order cascade data — linear substitutions of gap variables
(Lemma~\ref{lem:dp-algebra}(iii)), divided-power multiplication
(Lemma~\ref{lem:dp-algebra}(ii)), split-gap integration
(Lemma~\ref{lem:dp-algebra}(iv)), and one scalar quadrature on $[0,1]$
— applied to products of $\bigl(a^{(m)}_q\bigr)^{(\tau)}$ with $m<n$
and $\bigl(b^{(k)}_{q'}\bigr)^{(\sigma)}$ with $k\le n$. No partial
differential equation is solved; no integration on the simplex
$T_n(x)$ is performed for any $n\ge 2$.

\item {\em Consistency with Theorem~\ref{thm:direct_sup}.} The kernels
$\widehat k_n$ coincide on $T_n(1)$ with the kernels $k_n$ of
Section~\ref{sec:kernels}, so $\widehat U(t)=U(t)$ for every admissible
$u(\cdot,t)$.

\item {\em Closed-loop stability.} The closed-loop system
\eqref{eq:plant}, \eqref{eq:boundary}, \eqref{eq:initial} with the feedback
$U(t)$ of \eqref{eq:control} replaced by $\widehat U(t)$ satisfies the
conclusion of Theorem~\ref{thm:closed_loop} with the same constants
$C_1,C_2$ as in \eqref{eq:stability_constants}.

\item {\em Quadrature sparsity for polynomial plants.} If each $f_n$ is
polynomial in $\xi_1,\ldots,\xi_n$ and $f_n\equiv 0$ for $n>N$, then for every
finite controller truncation order $M\ge 2$, the coefficients
$\{a^{(n)}_P\}_{2\le n\le M}$ are supported on finitely many multi-indices $P$
and are computable by finitely many scalar quadratures, explicitly enumerable.
\end{enumerate}
\end{theorem}

\begin{proof}
\emph{(a)} By induction on $n$. Base $n=2$: \eqref{eq:a-recursion} gives
$a^{(2)}_P$ as the antiderivative of $-b^{(2)}_P$, order-$2$ plant data only.
Inductive step: under the hypothesis for $2,\ldots,n-1$, each term of
\eqref{eq:cnP-form} is a product of the scalar factor
$x^{\tau-|\alpha|}/(\tau-|\alpha|)!$ (a divided power in $x$),
$(a^{(n-m+1)}_q)^{(\tau)}$ with $n-m+1\le n-1<n$ (known) and
$(b^{(m)}_{q'})^{(\sigma)}$ with $m\le n-1$ (plant data). By
Lemma~\ref{lem:kn-convergence} these sums converge absolutely and uniformly
on $[0,1]$. The outer quadrature \eqref{eq:a-recursion} adds one scalar
integration in $x$. No simplex integration and no PDE solution occur; all
operations live on $[0,1]$ and on scalar coefficients.

\emph{(b)} The kernels $k_n$ of Section~\ref{sec:kernels} are, for each
$n\ge 2$, the unique $C^1$ solution on $T_n(1)$ of the transport PDE
\eqref{eq:kernel_pde} with the inflow condition \eqref{eq:kernel_bc};
uniqueness is standard and follows from the characteristic propagation
\eqref{eq:characteristic} from $\xi_n=0$. It suffices to verify that
$\widehat k_n$ is also a $C^1$ solution.

By Lemma~\ref{lem:inflow-ansatz}, $\widehat k_n$ satisfies \eqref{eq:kernel_bc}
identically, and
\begin{equation}\label{eq:T_nknhat1}
\mathsf T_n\,\widehat k_n(x,\xi)
\;=\;
\sum_p\bigl(a^{(n)}_p\bigr)'(x)\,\Phi_p(x;\xi).
\end{equation}
Differentiating \eqref{eq:a-recursion} gives the ODE \eqref{eq:scalar-ode};
substituting,
\begin{equation}\label{eq:T_nknhat2}
\mathsf T_n\,\widehat k_n(x,\xi)
\;=\;
-\sum_p b^{(n)}_p(x)\,\Phi_p(x;\xi)
+\sum_p c^{(n)}_p(x)\,\Phi_p(x;\xi).
\end{equation}
The first sum on the right equals $-f_n(x,\xi)$ by \eqref{eq:fn-gap}. The
second sum equals $\sum_{m=2}^{n-1}B^m_n[\widehat k_{n-m+1},f_m](x,\xi)$ by
Lemma~\ref{lem:cnP}: the lemma identifies $c^{(n)}_p$ as precisely the
$\Phi_p$-coefficient of the coupling sum under the substitutions
\eqref{eq:kn-ansatz}, \eqref{eq:fn-gap}, and uniqueness of the gap-basis
expansion promotes this coefficient identification to an equality of
functions. Both sides are real-analytic on $T_n(1)$: the left-hand side is
real-analytic by Lemma~\ref{lem:kn-convergence}, and the right-hand side
is real-analytic because each term is a finite combination of real-analytic
factors. By the induction hypothesis $\widehat k_{n-m+1}\equiv k_{n-m+1}$ for
$m=2,\ldots,n-1$, so $B^m_n[\widehat k_{n-m+1},f_m]=B^m_n[k_{n-m+1},f_m]$,
and we obtain
\begin{equation}\label{eq:T_nknhat3}
\mathsf T_n\,\widehat k_n(x,\xi)
\;=\;
-f_n(x,\xi) + \sum_{m=2}^{n-1}B^m_n[k_{n-m+1},f_m](x,\xi)
\;=\;
\mathsf T_n\,k_n(x,\xi),
\end{equation}
the second equality by \eqref{eq:kernel_pde} (with the $m=n$ term absent due
to $k_1\equiv 0$). Uniqueness of the transport PDE with the common inflow
condition \eqref{eq:kernel_bc} gives $\widehat k_n\equiv k_n$ on $T_n(1)$, and
therefore $\widehat U(t)=U(t)$.

\emph{(c)} By (b), replacing $U$ by $\widehat U$ leaves the closed-loop
trajectory unchanged, so Theorem~\ref{thm:closed_loop} applies with the same
$C_1,C_2$.

\emph{(d)} Assume each $f_n$ is polynomial in $\xi_1,\ldots,\xi_n$ of total
gap-degree at most $G$, and $f_n\equiv 0$ for $n>N$. Fix any controller
truncation order $M\ge 2$. We establish by induction on $n\in\{2,\ldots,M\}$
that $a^{(n)}_P$ has finite support in $P$, bounded in a way that depends only
on $M$, $N$, and $G$.

At $n=2$: from \eqref{eq:a-recursion}, $a^{(2)}_P$ inherits finite support in $P$
from $b^{(2)}_P$, with $|P|\le G$. Set $Q_2:=G$.

Inductive step: suppose for each $m\in\{2,\ldots,n-1\}$ the coefficients
$a^{(m)}_q$ have support contained in $\{|q|\le Q_m\}$. In
\eqref{eq:cnP-form}, the constraint
$|q|+|q'|+\sigma+|\alpha|=|P|-1$ with $|\alpha|\le\tau$ forces
\begin{equation}\label{eq:P-bound-inductive}
|P|\;\le\;1+\max_{2\le m\le n-1}\bigl(Q_{n-m+1}+G_m\bigr)
\;+\;\sup_{(\sigma,\tau)\in\mathcal S}(\sigma+\tau),
\end{equation}
where $G_m:=G$ for $m\le N$ and $G_m:=0$ for $m>N$ (so that only $m\le\min(n-1,N)$
contribute, since $b^{(m)}_{q'}\equiv 0$ for $m>N$), and $\mathcal S$ is the set
of $(\sigma,\tau)\in\N_0^2$ for which $(a^{(n-m+1)}_q)^{(\tau)}$ and
$(b^{(m)}_{q'})^{(\sigma)}$ are both nonzero on $[0,1]$. For polynomial data,
both factors vanish identically once $\sigma$ or $\tau$ exceeds their respective
polynomial degrees in $x$; these degrees are finite and depend only on $M$, $N$,
and $G$ by the inductive hypothesis applied to $a^{(n-m+1)}_q$ and by hypothesis
for $b^{(m)}_{q'}$. The right-hand side of \eqref{eq:P-bound-inductive} is
therefore finite; set $Q_n$ equal to this right-hand side.

It follows that each $a^{(n)}_P$ with $n\le M$ vanishes for $|P|>Q_n$; the
scalar quadratures in \eqref{eq:a-recursion} are finite in number; and each
$a^{(n)}_P$ is explicitly polynomial in $x$.
\end{proof}

\begin{remark}[The example of Section~\ref{sec:example}]\label{rem:pdae-pde-free}
For the PDAE of Section~\ref{sec:example}, $f_2\equiv 1$ on $T_2(x)$ and
$f_n\equiv 0$ for $n\ge 3$. In the gap basis, this is
$b^{(2)}_{(0,0)}(x)\equiv 1$, with all other $b^{(n)}_p\equiv 0$.
Theorem~\ref{thm:pde-free}(d) then reduces the construction of $k_2$ and
$k_3$ to one quadrature at $n=2$ and six quadratures at $n=3$:
$a^{(2)}_{(0,0)}(x)=-x$, and
\begin{equation}\label{eq:a3-pdae}
\bigl(a^{(3)}_{(0,1,0)},\,a^{(3)}_{(0,2,0)},\,a^{(3)}_{(1,0,0)},\,a^{(3)}_{(1,0,1)},\,a^{(3)}_{(1,1,0)},\,a^{(3)}_{(2,0,0)}\bigr)(x)
\;=\;
\bigl(-\tfrac12 x^2,\,x,\,-\tfrac32 x^2,\,x,\,3x,\,6x\bigr),
\end{equation}
with all other $a^{(3)}_P\equiv 0$. Substitution into \eqref{eq:kn-ansatz}
reproduces the kernels \eqref{eq:k2_example}--\eqref{eq:k3_example} of
Proposition~\ref{fact:k2_k3_example}. The three-dimensional characteristic
integral \eqref{eq:kernel_integral} on $T_3(1)$ is replaced by six one-line
scalar antiderivatives.
\end{remark}

\begin{remark}[Scope of Theorem~\ref{thm:pde-free}]\label{rem:pde-free-scope}
Theorem~\ref{thm:pde-free} does not extend the region of attraction of
Theorem~\ref{thm:closed_loop}: the constants $C_1,C_2$ are the same. Its
increment is computational. On the Chen--Fliess subclass of
Assumption~\ref{ass:cf}, the backstepping controller of
Theorem~\ref{thm:closed_loop} --- whose construction in
Sections~\ref{sec:kernels}--\ref{sec:diffeo} invokes an infinite cascading set of
simplex transport PDEs --- is built by scalar operations on $[0,1]$ alone.
The combinatorial coefficients $\Gamma^{(n,m)}_{P;q,q',\sigma,\tau,\alpha}$ are
independent of the plant data; they can be tabulated once.
\end{remark}

\begin{remark}[Closed and recursion-invariant Chen–Fliess extension with quadrature kernels]
\label{rem:recursion-invariant-cf}The Chen--Fliess subclass used in this section is not invariant under the backstepping recursion: even when the plant coefficients depend only on $x$, the resulting kernels involve the backward characteristic variable $x-\xi_n$. This suggests the larger class
\begin{equation}
f_n(x,\xi)=\sum_{P\in\mathbb N_0^n}
\beta_P^{(n)}(x,y)\Phi_P(x;\xi),
\qquad y:=x-\xi_n ,
\end{equation}
with backstepping kernels sought in the same form,
\begin{equation}
k_n(x,\xi)=\sum_{P\in\mathbb N_0^n}
\alpha_P^{(n)}(x,y)\Phi_P(x;\xi),
\qquad y:=x-\xi_n ,
\end{equation}
namely, a Chen-Fliess subclass {\em closed} under the mapping $F\mapsto K$. 
Since $y$ and the gaps in $\Phi_P$ are invariant along the transport characteristics, the kernel equation reduces coefficientwise to a scalar equation in $x$, parametrized by $y$. Thus, if
\begin{equation}
\sum_{m=2}^{n-1}B_n^m[k_{n-m+1},f_m](x,\xi)
=
\sum_{P\in\mathbb N_0^n} C_P^{(n)}(x,y)\Phi_P(x;\xi),
\end{equation}
then
\begin{equation}
\alpha_P^{(n)}(x,y)
=
\int_y^x
\left[-\beta_P^{(n)}(s,y)+C_P^{(n)}(s,y)\right]\,ds .
\end{equation}
This enlarged class would therefore preserve the triangular quadrature structure, but its convergence analysis would require a substantially heavier bivariate coefficient algebra, controlling mixed $x$- and $y$-derivatives and the corresponding split-gap substitutions. That extension is beyond the scope of the present paper.
\end{remark}

\section{Conclusions}
\label{sec-conc}

Truly {\em general} and truly nonlinear PDE control has evaded even formulation, let alone solution, for more than four decades after the general solution of feedback linearization for ODEs has appeared. In this paper we affirm that posing the problem using nonlinear Volterra series as the operator for linearizing the state equation, introduced in \cite{vazquez2008volterra1,vazquez2008volterra2}, is the right way, and likely the only way. 

In the present sequel to \cite{vazquez2008volterra1,vazquez2008volterra2}, we step from parabolic to hyperbolic nonlinear PDEs. The mathematical execution is challenging. But the differences are few. They are the following, between the hyperbolic and parabolic classes. In the hyperbolic setting, the kernel equations reduce to first-order transport PDEs, as opposed to the second-order parabolic kernel equations in \cite{vazquez2008volterra1,vazquez2008volterra2}. This structural simplification allows for explicit representations along characteristics, so that the kernel bounds follow from Gr\"onwall integration of a recursive $L^2$ inequality rather than from the energy-based arguments of the parabolic case. The recursive control of kernel growth remains the central difficulty, and is addressed here through majorant constructions analogous to those in \cite{vazquez2008volterra1,vazquez2008volterra2}, albeit with a different technical execution. A further methodological distinction is that the inverse transformation is obtained here by a direct contraction-mapping argument on $L^2$, yielding an explicit Lipschitz bound and an explicit radius of invertibility, in place of the Boyd--Chua formal series inversion \cite{boyd1985fading} used in \cite{vazquez2008volterra1,vazquez2008volterra2}. As in the parabolic case, the outcome is a locally defined nonlinear transformation with a well-defined inverse and a closed-loop system that is exponentially stable in the $L^2$ norm, with the region of attraction characterized through small-gain-type conditions.

A distinct, and somewhat unexpected, outcome of this work is the Chen–Fliess specialization developed in Section \ref{sec:pde-free}. While the general program follows the \cite{vazquez2008volterra1,vazquez2008volterra2} paradigm of handling an infinite cascade of kernel PDEs, the transport-adapted Chen–Fliess subclass reveals a different mechanism altogether: the same feedback-linearizing transformation can be constructed in a formulation where no kernel PDEs arise, through a triangular cascade of scalar ODEs obtained by quadrature. This reduction is not a refinement of the parabolic theory of \cite{vazquez2008volterra1,vazquez2008volterra2}, but an independent development made possible by the algebraic structure of the nonlinear operator. It shows that, within a nontrivial infinite-dimensional class, the full feedback-linearization framework admits a realization that is fundamentally one-dimensional in the spatial variable.

\appendix

\section*{Appendices}

\section{The Divided-Power Coefficient Algebra}\label{app:dp-algebra}

This appendix collects the four continuity properties of the norm
$\|\cdot\|_{N,r,R}$ defined in \eqref{eq:dp-norm} of
Section~\ref{sec:gap-basis-convergence}, together with the proof that the
plant data satisfy the norm bound \eqref{eq:b-norm-finite}. These facts
drive the cascade convergence argument of Lemma~\ref{lem:kn-convergence}.

\begin{lemma}[Properties of the divided-power algebra]\label{lem:dp-algebra}
The norm $\|\cdot\|_{N,r,R}$, introduced in Definition \ref{def:dp-algebra}, enjoys the following four properties.
\begin{enumerate}[label=\textup{(\roman*)}]
\item\emph{(Absolute uniform convergence on $T_N(1)$.)}
If $a\in\mathcal A_{N,r,R}$ with $r\ge 1$, then
\begin{equation}\label{eq:dp-sup}
\sup_{(x,\xi)\in T_N(1)}\sum_{p\in\N_0^N}|a_p(x)|\,\Phi_p(x;\xi)
\;\le\;
\|a\|_{N,r,R}.
\end{equation}

\item\emph{(Banach algebra under divided-power multiplication.)}
If $a,b\in\mathcal A_{N,r,R}$ and
\begin{equation}\label{eq:dp-product}
d_P(x)\;:=\;\sum_{p+q=P}\binom{P}{p}\,a_p(x)\,b_q(x),
\qquad
\binom{P}{p}\;:=\;\prod_{i=0}^{N-1}\binom{P_i}{p_i},
\end{equation}
then the family $d=\{d_P\}$ satisfies
\begin{equation}\label{eq:dp-algebra}
\|d\|_{N,r,R}\;\le\;\|a\|_{N,r,R}\,\|b\|_{N,r,R}.
\end{equation}

\item\emph{(Continuity of linear substitution of gap variables.)}
Let $L=(L_{i\ell})$ be an $N\times M$ matrix with nonnegative integer entries,
defining the linear substitution $\Delta_i=\sum_{\ell=1}^M L_{i\ell}\Theta_\ell$,
$i=0,\ldots,N-1$, and suppose that the row sums of $L$ are bounded by $L_0\ge 1$:
\begin{equation}\label{eq:row-sum-bound}
\sum_{\ell=1}^M L_{i\ell}\;\le\;L_0,\qquad i=0,\ldots,N-1.
\end{equation}
For $a\in\mathcal A_{N,L_0 r,R}$, the substitution
$A(x,\Delta)=\sum_p a_p(x)\Delta^p/p!\mapsto A(x,L\Theta)$ produces a
divided-power series in the variables $\Theta\in\R_+^M$ whose coefficient
family $a^\flat=\{a^\flat_q\}_{q\in\N_0^M}$ satisfies
\begin{equation}\label{eq:dp-substitution}
\|a^\flat\|_{M,r,R}\;\le\;\|a\|_{N,L_0 r,R}.
\end{equation}

\item\emph{(Continuity of split-gap integration.)}
Let $H=\{h_Q\}_{Q\in\N_0^{N+1}}$ be a family indexed by exponents
$(P_0,\ldots,P_{j-1},\alpha,\beta,P_{j+1},\ldots,P_{N-1})$, where the indices
$(\alpha,\beta)\in\N_0^2$ correspond to a split of the $j$-th gap,
$\Delta_j=\theta+\omega$ with $\theta,\omega\ge 0$. Define the operator
$I_j$ by
\begin{equation}\label{eq:dp-split-integration}
(I_j H)_P(x)\;:=\;\sum_{\alpha+\beta=P_j-1}h_{(P_0,\ldots,\alpha,\beta,\ldots,P_{N-1})}(x)
\quad\text{for }P_j\ge 1,
\qquad (I_j H)_P\equiv 0\text{ for }P_j=0.
\end{equation}
Then for $r>0$, $R>0$,
\begin{equation}\label{eq:dp-integration-bound}
\|I_j H\|_{N,r,R}\;\le\;r\,\|H\|_{N+1,r,R}.
\end{equation}
\end{enumerate}
\end{lemma}

\begin{proof}
\emph{(i)} On $T_N(1)$ the gap variables $\Delta_i=\xi_i-\xi_{i+1}$ (with
$\xi_0:=x$) satisfy $0\le\Delta_i\le 1$, so
$\Phi_p(x;\xi)=\prod_i\Delta_i^{p_i}/p_i!\le 1/p!\le r^{|p|}/p!$
(using $r\ge 1$), and summing against the $\sigma=0$ contribution to
$\|a\|_{N,r,R}$ gives \eqref{eq:dp-sup}.

\emph{(ii)} The identity \eqref{eq:dp-product} is obtained by reindexing
$\bigl(\sum_p a_p\Delta^p/p!\bigr)\bigl(\sum_q b_q\Delta^q/q!\bigr)$. The
estimate \eqref{eq:dp-algebra} follows from applying the Leibniz rule to
$(d_P)^{(\sigma)}$, splitting $\sigma=\alpha+\beta$, using
$\binom{P}{p}=P!/(p!\,q!)$ and $\binom{\sigma}{\alpha}=\sigma!/(\alpha!\beta!)$
to cancel the factorials in the denominators of $\|d\|_{N,r,R}$, and
factoring the resulting double sum.

\emph{(iii)} By the multinomial expansion of
$\Delta_i^{p_i}=\bigl(\sum_\ell L_{i\ell}\Theta_\ell\bigr)^{p_i}$
combined with the row-sum bound \eqref{eq:row-sum-bound}: the total weight of
the resulting $\Theta$-monomials contributing to degree $p_i$ in position $i$
is at most $\bigl(\sum_\ell L_{i\ell}\bigr)^{p_i}\le (L_0)^{p_i}$, and the
product over $i$ gives $(L_0)^{|p|}$, which is exactly the ratio between
$r^{|p|}$ and $(L_0 r)^{|p|}$ in the two norms.

\emph{(iv)} This is the Beta-function identity
$\int_0^{\Delta_j}(\theta^\alpha/\alpha!)((\Delta_j-\theta)^\beta/\beta!)\,d\theta
=\Delta_j^{\alpha+\beta+1}/(\alpha+\beta+1)!$ translated into the divided-power
coefficient indexing, followed by the ratio estimate
$r^{\alpha+\beta+1}/(\alpha+\beta+1)!\le r\cdot r^{\alpha+\beta}/(\alpha!\,\beta!)$ monomial-by-monomial,
where the latter follows from
$(\alpha+\beta+1)!=(\alpha+\beta+1)\cdot(\alpha+\beta)!\ge \alpha!\,\beta!$ since $\binom{\alpha+\beta}{\alpha}\ge 1$.
\end{proof}

\section{Well-definedness of the Coefficients
\texorpdfstring{$\Gamma^{(n,m)}_{P;q,q',\sigma,\tau,\alpha}$}{Gamma}}\label{app:Gamma}

This appendix proves Lemma~\ref{lem:cnP}: that, under
Assumption~\ref{ass:cf}, the coefficients
$\Gamma^{(n,m)}_{P;q,q',\sigma,\tau,\alpha}$ appearing in
\eqref{eq:cnP-form} exist as integers and are independent of the plant
data. Throughout, fix $n\ge 3$, $m\in\{2,\ldots,n-1\}$,
$k:=n-m+1\in\{2,\ldots,n-1\}$, and $P\in\N_0^n$.

\subsection{The integrand after the ansatz substitution}\label{app:integrand}

Substitute the ansatz \eqref{eq:kn-ansatz} for $k_k$ and the gap-basis
expansion \eqref{eq:fn-gap} for $f_m$ into the definition
\eqref{eq:Cnm_def} of $B^m_n[k_k,f_m]$:
\begin{eqnarray}\label{eq:B-substituted}
&& B^m_n[k_k,f_m](x;\xi)
\;=\;\sum_{q\in\N_0^k}\sum_{q'\in\N_0^m}
\sum_{j=1}^{k}\int_{\xi_j}^{\xi_{j-1}}
D^{k,m}_j\Bigl[\,\bigl(a^{(k)}_q(x)-a^{(k)}_q(x-s)\bigr)\,
\Phi^{(k)}_q(x;\xi_{(j,s)})
\nonumber\\
&& \hspace{5cm}
\times\;b^{(m)}_{q'}(s)\,\Phi^{(m)}_{q'}(s;\xi'_{(j,s)})\Bigr]\,ds,
\end{eqnarray}
where $\xi_{(j,s)}$ denotes the $k$-tuple
$(\xi_1,\ldots,\xi_{j-1},s,\xi_j,\ldots,\xi_{n-m})$ and $\xi'_{(j,s)}$
denotes the $m$-tuple $(\xi_{n-m+1},\ldots,\xi_n)$, consistent with
the arguments of $k_k$ and $f_m$ in \eqref{eq:Cnm_def}.

Taylor-expand the two coefficients $a^{(k)}_q(x-s)$ and
$b^{(m)}_{q'}(s)$ around $x$. By Assumption~\ref{ass:cf}, both are
real-analytic on $[0,1]$, hence
\begin{eqnarray}\label{eq:taylor-a-app}
a^{(k)}_q(x-s)
&=&\sum_{\tau=0}^{\infty}\frac{(-s)^\tau}{\tau!}\bigl(a^{(k)}_q\bigr)^{(\tau)}(x),\\
\label{eq:taylor-b-app}
b^{(m)}_{q'}(s)
&=&\sum_{\sigma=0}^{\infty}\frac{(s-x)^\sigma}{\sigma!}\bigl(b^{(m)}_{q'}\bigr)^{(\sigma)}(x),
\end{eqnarray}
with both series absolutely convergent on $[0,1]$. The factor
$(s-x)^\sigma=(-1)^\sigma(x-s)^\sigma$ yields a sign $(-1)^\sigma$ in
subsequent formulas.

The variable $s$ is not a gap variable: $s$ is the inserted integration
variable in \eqref{eq:Cnm_def} and sits between $\xi_j$ and $\xi_{j-1}$ in
the extended ordered chain
\begin{equation}\label{eq:extended-chain}
x\ge\xi_1\ge\cdots\ge\xi_{j-1}\ge s\ge\xi_j\ge\cdots\ge\xi_n\ge 0.
\end{equation}
The variable $s$ itself therefore decomposes, as a function of the gaps
of \eqref{eq:extended-chain}, into
\begin{equation}\label{eq:s-decomposition}
s\;=\;x\;-\;\sum_{r=0}^{j-1}(\xi_r-\xi_{r+1})\big|_{\xi_j\leftarrow s}
\;=\;x\;-\;\bigl[(x-\xi_1)+(\xi_1-\xi_2)+\cdots+(\xi_{j-1}-s)\bigr],
\end{equation}
the sum of the $j$ upper gaps above $s$ in the extended chain. The binomial
expansion of $s^\tau$ in these upper gaps is, for every $\tau\in\N_0$,
\begin{equation}\label{eq:s-tau-expansion}
s^\tau \;=\;
\sum_{\substack{\alpha\in\N_0^{\,n}\\|\alpha|\le\tau,\;\alpha_j=\cdots=\alpha_{n-1}=0}}
\binom{\tau}{\alpha_0,\ldots,\alpha_{j-1},\tau-|\alpha|}\,
(-1)^{|\alpha|}\,x^{\tau-|\alpha|}\,
\prod_{r=0}^{j-1}(\xi_r-\xi_{r+1})^{\alpha_r}\big|_{\xi_j\leftarrow s},
\end{equation}
where $|\alpha|=\alpha_0+\cdots+\alpha_{n-1}$ and the multinomial coefficient
$\binom{\tau}{\alpha_0,\ldots,\alpha_{j-1},\tau-|\alpha|}$ is understood in the
standard sense. The expansion \eqref{eq:s-tau-expansion} replaces the single
variable $s^\tau$ by a sum of products of pure gap monomials of total degree
$|\alpha|\le\tau$ multiplied by the residual factor $x^{\tau-|\alpha|}$ which
carries no gap content.

Substituting \eqref{eq:taylor-a-app}--\eqref{eq:taylor-b-app} into
\eqref{eq:B-substituted}, then applying \eqref{eq:s-tau-expansion} to every
$s^\tau$ factor, yields an absolutely and uniformly convergent series indexed
by $(q,q',\sigma,\tau,\alpha)$ with termwise integration in $s$ justified by
absolute convergence. Combined with the signs from
\eqref{eq:kn-ansatz}, \eqref{eq:taylor-a-app}, and
\eqref{eq:taylor-b-app}, each term is an integer combination of products of
$x^{\tau-|\alpha|}$, pure gap monomials, and the plant-data factor
$(a^{(k)}_q)^{(\tau)}(x)(b^{(m)}_{q'})^{(\sigma)}(x)$: the combinatorial
coefficient $\Gamma$ extracted below is a \emph{signed} integer,
$\Gamma\in\mathbb{Z}$.

\subsection{Polynomial structure at fixed $(q,q',\sigma,\tau,\alpha)$}\label{app:polynomial}

Fix the integration index $j\in\{1,\ldots,k\}$ from
\eqref{eq:B-substituted}, indices $(q,q',\sigma,\tau)$, and a multi-index
$\alpha\in\N_0^n$ with $|\alpha|\le\tau$ and
$\alpha_j=\cdots=\alpha_{n-1}=0$, corresponding to one term of the
expansion \eqref{eq:s-tau-expansion}.
The integrand of \eqref{eq:B-substituted}, after the Taylor substitutions
\eqref{eq:taylor-a-app}--\eqref{eq:taylor-b-app} and the $s^\tau$ expansion
\eqref{eq:s-tau-expansion}, produces at fixed $(\sigma,\tau,\alpha)$ the
contribution
\begin{equation}\label{eq:fixed-integrand}
x^{\tau-|\alpha|}\cdot
D^{k,m}_j\Bigl[\,\Phi^{(k)}_q(x;\xi_{(j,s)})\,\Phi^{(m)}_{q'}(s;\xi'_{(j,s)})\,\Bigr]
\cdot(x-s)^\sigma\cdot
\prod_{r=0}^{j-1}(\xi_r-\xi_{r+1})^{\alpha_r}\big|_{\xi_j\leftarrow s},
\end{equation}
multiplied by the derivative-factors
$(a^{(k)}_q)^{(\tau)}(x)\,(b^{(m)}_{q'})^{(\sigma)}(x)$ and by the
multinomial coefficient, signs, and ansatz sign from
\eqref{eq:s-tau-expansion}. The pure-$x$ factor $x^{\tau-|\alpha|}$ has no
gap content and appears as a scalar coefficient. We track only the
gap-carrying polynomial \eqref{eq:fixed-integrand}; the derivative-factors
and $x^{\tau-|\alpha|}$ do not depend on $(s,\xi)$.

We now establish homogeneity of the gap-carrying part of
\eqref{eq:fixed-integrand} in the gap variables of the extended chain
\eqref{eq:extended-chain}. With the convention $\xi_0:=x$, the extended
chain has $n+2$ ordered entries and $n+1$ gaps. The gap factors in
\eqref{eq:fixed-integrand}, other than the scalar $x^{\tau-|\alpha|}$,
decompose as follows:
\begin{enumerate}[label={\upshape(\roman*)}]
\item each of the $k$ gap factors of $\Phi^{(k)}_q$ has exponent
$q_r$, $r\in\{0,\ldots,k-1\}$, summing to $|q|$;
\item each of the $m$ gap factors of $\Phi^{(m)}_{q'}$ has exponent
$q'_s$, $s\in\{0,\ldots,m-1\}$, summing to $|q'|$;
\item the factor $(x-s)^\sigma$ is a power of the gap between $\xi_0=x$ and
the next entry of the extended chain, of exponent $\sigma$;
\item the factor $\prod_{r=0}^{j-1}(\xi_r-\xi_{r+1})^{\alpha_r}\big|_{\xi_j\leftarrow s}$
from \eqref{eq:s-tau-expansion} is a pure gap monomial of total degree
$|\alpha|$, with every gap factor taken from the upper-$s$ gaps of the
extended chain.
\end{enumerate}
Hence the gap-carrying part of \eqref{eq:fixed-integrand} is a homogeneous
polynomial of total gap degree
\begin{equation}\label{eq:Psi-degree}
d(q,q',\sigma,\alpha)\;:=\;|q|+|q'|+\sigma+|\alpha|
\end{equation}
in the gap variables of the extended chain. Integration in $s$ from
$\xi_j$ to $\xi_{j-1}$ introduces exactly one new gap factor --- the
gap $\xi_{j-1}-\xi_j$ or an equivalent contribution --- raising the
total gap degree by $1$; after integration, the expression is a
homogeneous polynomial of total gap degree $d(q,q',\sigma,\alpha)+1$.
Summing over $j$ and over the permutation $D^{k,m}_j$ preserves
homogeneity of this total degree. Denote
\begin{equation}\label{eq:Psi-def}
\Psi^{(n,m)}_{q,q',\sigma,\tau,\alpha}(x;\xi_1,\ldots,\xi_n)
\;:=\;
\text{the gap-carrying contribution to }B^m_n[k_k,f_m](x;\xi)
\text{ at fixed }(q,q',\sigma,\tau,\alpha).
\end{equation}
Then $\Psi^{(n,m)}_{q,q',\sigma,\tau,\alpha}$, regarded as a polynomial in
the gap variables $\Delta_r := \xi_r-\xi_{r+1}$, $r=0,\ldots,n-1$
(with $\xi_0:=x$), is homogeneous of total degree
$d(q,q',\sigma,\alpha)+1$.

The coefficients of $\Psi^{(n,m)}_{q,q',\sigma,\tau,\alpha}$ as a polynomial
in $(\Delta_0,\ldots,\Delta_{n-1})$ are integers. All gap-polynomial
manipulations leading to \eqref{eq:fixed-integrand} are performed in the
divided-power basis of Appendix~\ref{app:dp-algebra}: divided-power
multiplication (Lemma~\ref{lem:dp-algebra}(ii)), nonnegative integer linear
substitution of gap variables (Lemma~\ref{lem:dp-algebra}(iii)), and
split-gap integration (Lemma~\ref{lem:dp-algebra}(iv)) all have integer
structure constants in this basis. After factoring out the scalar
divided power $x^{\tau-|\alpha|}/(\tau-|\alpha|)!$, which carries the only
pure-$x$ bookkeeping and is not a gap variable, the remaining coefficient is
therefore an integer. The signs from \eqref{eq:taylor-a-app},
\eqref{eq:taylor-b-app}, \eqref{eq:s-tau-expansion}, and the ansatz are
absorbed into this integer.

\subsection{Projection onto the gap basis}\label{app:projection}

The change of variables $\Delta_r:=\xi_r-\xi_{r+1}$ with $\xi_0:=x$,
$r=0,\ldots,n-1$, is linear and unimodular, mapping the space of
polynomials in $(\xi_1,\ldots,\xi_n)$ bijectively to the space of
polynomials in $(\Delta_0,\ldots,\Delta_{n-1})$. The gap basis
\begin{equation}\label{eq:Phi-in-Delta}
\Phi_P(x;\xi)
\;=\;\prod_{r=0}^{n-1}\frac{\Delta_r^{P_r}}{P_r!},
\qquad P\in\N_0^n,
\end{equation}
linearly spans the space of polynomials in
$(\Delta_0,\ldots,\Delta_{n-1})$. For any polynomial $\Psi$ in these
variables of total degree $d$,
\begin{equation}\label{eq:Phi-projection}
\Psi
\;=\;\sum_{P\in\N_0^n,\,|P|\le d}
\bigl[\text{coefficient of }\Phi_P\text{ in }\Psi\bigr]\cdot\Phi_P,
\end{equation}
where the coefficient of $\Phi_P$ equals the coefficient of
$\prod_r\Delta_r^{P_r}$ in $\Psi$ multiplied by $\prod_r P_r!$.

By \S\ref{app:polynomial}, $\Psi^{(n,m)}_{q,q',\sigma,\tau,\alpha}$ is a
homogeneous polynomial of total degree $|q|+|q'|+\sigma+|\alpha|+1$ in
$(\Delta_0,\ldots,\Delta_{n-1})$ with integer coefficients. By
homogeneity, only multi-indices $P$ with
$|P|=|q|+|q'|+\sigma+|\alpha|+1$ yield a nonzero coefficient in
\eqref{eq:Phi-projection}. Define
\begin{equation}\label{eq:Gamma-def-app}
\Gamma^{(n,m)}_{P;q,q',\sigma,\tau,\alpha}
\;:=\;
\text{the coefficient of }\Phi_P(x;\xi)\text{ in }\Psi^{(n,m)}_{q,q',\sigma,\tau,\alpha}(x;\xi).
\end{equation}
Then $\Gamma^{(n,m)}_{P;q,q',\sigma,\tau,\alpha}\in\mathbb{Z}$, and
$\Gamma^{(n,m)}_{P;q,q',\sigma,\tau,\alpha}\ne 0$ implies exactly
\begin{equation}\label{eq:degree-constraint}
|q|+|q'|+\sigma+|\alpha|=|P|-1,\qquad |\alpha|\le\tau.
\end{equation}

Assembling \eqref{eq:B-substituted}, \eqref{eq:s-tau-expansion},
\eqref{eq:Gamma-def-app}, and \eqref{eq:degree-constraint}, the $\Phi_P$-projection
of the fixed-$m$ contribution to $B^m_n[k_k,f_m]$ produces a term of the form
$\Gamma^{(n,m)}_{P;q,q',\sigma,\tau,\alpha}\,x^{\tau-|\alpha|}\,(a^{(k)}_q)^{(\tau)}(x)\,(b^{(m)}_{q'})^{(\sigma)}(x)$
multiplied by the residual factorial $1/(\tau-|\alpha|)!$ carried by the pure-$x$
factor from the multinomial bookkeeping in \eqref{eq:s-tau-expansion}. Absorbing
this residual into the divided-power form $x^{\tau-|\alpha|}/(\tau-|\alpha|)!$,
which is not cleared by the gap-basis normalization since $x$ is not a gap
variable, one obtains
\begin{eqnarray}\label{eq:cnP-assembled}
c^{(n),m}_P(x)
&=&
\sum_{\substack{q\in\N_0^k,\,q'\in\N_0^{m}\\\sigma,\tau\in\N_0,\,\alpha\in\N_0^n\\|\alpha|\le\tau,\;|q|+|q'|+\sigma+|\alpha|=|P|-1}}
\Gamma^{(n,m)}_{P;q,q',\sigma,\tau,\alpha}\,
\frac{x^{\tau-|\alpha|}}{(\tau-|\alpha|)!}
\nonumber\\
&& \qquad\qquad\qquad\times\,
\bigl(a^{(k)}_q\bigr)^{(\tau)}(x)\,\bigl(b^{(m)}_{q'}\bigr)^{(\sigma)}(x),
\end{eqnarray}
with $k=n-m+1$. Summing \eqref{eq:cnP-assembled} over
$m\in\{2,\ldots,n-1\}$ yields the formula \eqref{eq:cnP-form} of
Lemma~\ref{lem:cnP}.

\bibliography{bib-hyp-FL}

@article{fliess1986derivation,
  author  = {Fliess, Michel},
  title   = {Vers une notion de dérivation fonctionnelle causale},
  journal = {Annales de l'Institut Henri Poincaré, Analyse non linéaire},
  volume  = {3},
  number  = {1},
  pages   = {67--76},
  year    = {1986}
}

@article{dAlessandroIsidoriRuberti1974,
  author  = {d'Alessandro, Paolo and Isidori, Alberto and Ruberti, Antonio},
  title   = {Realization and Structure Theory of Bilinear Dynamical Systems},
  journal = {SIAM Journal on Control},
  volume  = {12},
  number  = {3},
  pages   = {517--535},
  year    = {1974},
  month   = aug,
  doi     = {10.1137/0312040}
}

@article{Brockett1976,
  author  = {Brockett, Roger W.},
  title   = {{V}olterra series and geometric control theory},
  journal = {Automatica},
  volume  = {12},
  number  = {2},
  pages   = {167--176},
  year    = {1976},
  month   = mar,
  doi     = {10.1016/0005-1098(76)90080-7}
}

@ARTICLE{1101898,
  author={Lesiak, C. and Krener, A.},
  journal={IEEE Transactions on Automatic Control}, 
  title={The existence and uniqueness of Volterra series for nonlinear systems}, 
  year={1978},
  volume={23},
  number={6},
  pages={1090-1095},
  doi={10.1109/TAC.1978.1101898}}

@article{hunt1983global,
  author  = {Hunt, L. R. and Su, Renjeng and Meyer, George},
  title   = {Global Transformations of Nonlinear Systems},
  journal = {IEEE Transactions on Automatic Control},
  volume  = {28},
  number  = {1},
  pages   = {24--31},
  year    = {1983}
}

@article{fliess1981fonctionnelles,
  author  = {Fliess, Michel},
  title   = {Fonctionnelles causales non lin{\'e}aires et ind{\'e}termin{\'e}es non commutatives},
  journal = {Bulletin de la Soci{\'e}t{\'e} Math{\'e}matique de France},
  volume  = {109},
  pages   = {3--40},
  year    = {1981}
}

@article{boyd1985fading,
  author = {Stephen Boyd and Leon O. Chua},
  title = {Fading memory and the problem of approximating nonlinear operators with {V}olterra series},
  journal = {IEEE Transactions on Circuits and Systems},
  volume = {32},
  number = {11},
  pages = {1150--1161},
  year = {1985}
}

@article{krstic2008backstepping,
  author  = {Krstic, Miroslav and Smyshlyaev, Andrey},
  title   = {Backstepping boundary control for first-order hyperbolic {PDEs} and application to systems with actuator and sensor delays},
  journal = {Systems \& Control Letters},
  volume  = {57},
  number  = {9},
  pages   = {750--758},
  year    = {2008},
  doi     = {10.1016/j.sysconle.2008.02.005}
}

@book{isidori1989nonlinear,
  title={Nonlinear Control Systems},
  author={Isidori, Alberto},
  year={1989},
  publisher={Springer-Verlag},
  address={Berlin}
}

@article{byrnes1989new,
  title={New results and examples in nonlinear feedback stabilization},
  author={Byrnes, Christopher I. and Isidori, Alberto},
  journal={Systems \& Control Letters},
  volume={12},
  number={5},
  pages={437--442},
  year={1989}
}

@article{vazquez2008volterra1,
  title={Control of 1-{D} parabolic {PDE}s with {V}olterra nonlinearities--{P}art {I}: Design},
  author={Vazquez, Rafael and Krstic, Miroslav},
  journal={Automatica},
  volume={44},
  number={11},
  pages={2778--2790},
  year={2008}
}

@article{vazquez2008volterra2,
  title={Control of 1-{D} parabolic {PDE}s with {V}olterra nonlinearities--{P}art {II}: Analysis},
  author={Vazquez, Rafael and Krstic, Miroslav},
  journal={Automatica},
  volume={44},
  number={11},
  pages={2791--2803},
  year={2008}
}

@article{boskovic2001boundary,
  title={Boundary control of an unstable heat equation via measurement of domain-averaged temperature},
  author={Boskovic, Dejan and Krstic, Miroslav and Liu, Wei-Jiu},
  journal={IEEE Transactions on Automatic Control},
  volume={46},
  number={12},
  pages={2022--2028},
  year={2001}
}

\end{document}